\documentclass[a4paper,UKenglish,cleveref, autoref, thm-restate]{lipics-v2021}
\newcommand{\R}{\mathbb R}

\definecolor{darkred}{rgb}{1, 0.1, 0.3}
\definecolor{darkblue}{rgb}{0, 0, 1.0}

\newcommand{\dW}        {{\mathrm{d^{per}_{W,1}}}}
\newcommand{\dOT}       {{\mathrm{d_{OT}}}}

\newcommand{\perD}      {{\mathsf{P}}} %{{P}}
\newcommand{\anotherperD}   {{\mathsf{Q}}} %{{Q}}
\newcommand{\aL}        {{\mathcal{L}}}
\newcommand{\XX}        {{\mathrm{X}}}
\newcommand{\reals}     {{\mathbb{R}}}
\newcommand{\dgm}       {{\mathsf{dgm}}} 
\newcommand{\myaug}     {{\mathsf{aug}}}
\newcommand{\eps}       {{\varepsilon}}
\newcommand{\mysupp}    {{\mathsf{supp}}}
\newcommand{\aT}        {{\mathrm{T}}}
\newcommand{\aV}        {{\mathrm{V}}}
\newcommand{\Gflow}     {{f^*_G}}
\newcommand{\dTper}     {{\widehat{d}_{L_1}}}
\newcommand{\greedyM}   {{\widehat{\Gamma}}}
\newcommand{\dOTtree}   {{\mathrm{d_{OT,d_T}}}}
\newcommand{\dOTflow}   {{\mathrm{d_{OT}^{flow}}}}

\title{Approximation algorithms for 1-Wasserstein distance between persistence diagrams} %TODO Please add

\titlerunning{Approx. Algorithms for $W_1$ distance} %TODO optional, please use if title is longer than one line

\author{Samantha Chen}{University of California - San Diego, USA  }{sac003@ucsd.edu}{}{}%TODO mandatory, please use full name; only 1 author per \author macro; first two parameters are mandatory, other parameters can be empty. Please provide at least the name of the affiliation and the country. The full address is optional

\author{Yusu Wang}{University of California - San Diego, USA}{yusuwang@ucsd.edu}{}{}

\authorrunning{S. Chen and Y. Wang} %TODO mandatory. First: Use abbreviated first/middle names. Second (only in severe cases): Use first author plus 'et al.'

\Copyright{Samantha Chen and Yusu Wang} %TODO mandatory, please use full first names. LIPIcs license is "CC-BY";  http://creativecommons.org/licenses/by/3.0/

 %TODO mandatory: Please choose ACM 2012 classifications from https://dl.acm.org/ccs/ccs_flat.cfm 

 %TODO mandatory; please add comma-separated list of keywords

\category{} %optional, e.g. invited paper

\relatedversion{} %optional, e.g. full version hosted on arXiv, HAL, or other respository/website
\supplement{https://github.com/chens5/w1estimators.git}%optional, e.g. related research data, source code, ... hosted on a repository like zenodo, figshare, GitHub, ...
\funding{This work is partially supported by National Science Foundation (NSF) via grants OAC-2039794 and IIS- 2050360.}%optional, to capture a funding statement, which applies to all authors. Please enter author specific funding statements as fifth argument of the \author macro.

\acknowledgements{We want to thank Chen Cai for providing the reddit-binary and ModelNet10 datasets used in the experiments. }%optional

\nolinenumbers %uncomment to disable line numbering

\hideLIPIcs  %uncomment to remove references to LIPIcs series (logo, DOI, ...), e.g. when preparing a pre-final version to be uploaded to arXiv or another public repository

\EventEditors{David Coudert and Emanuele Natale}
\EventNoEds{2}
\EventLongTitle{19th International Symposium on Experimental Algorithms (SEA 2021)}
\EventShortTitle{SEA 2021}
\EventAcronym{SEA}
\EventYear{2021}
\EventDate{June 7--9, 2021}
\EventLocation{Nice, France}
\EventLogo{}
\SeriesVolume{190}
\ArticleNo{14}
\begin{document}

\maketitle

\begin{abstract}
Recent years have witnessed a tremendous growth using topological summaries, especially the persistence diagrams (encoding the so-called persistent homology) for analyzing complex shapes. 
Intuitively, persistent homology maps a potentially complex input object (be it a graph, an image, or a point set and so on) to a unified type of feature summary, called the persistence diagrams. 
One can then carry out downstream data analysis tasks using such persistence diagram representations. 
A key problem is to compute the distance between two persistence diagrams efficiently. 
In particular, a persistence diagram is essentially a multiset of points in the plane, and one popular distance is the so-called 1-Wasserstein distance between persistence diagrams.
In this paper, we present two algorithms to approximate the 1-Wasserstein distance for persistence diagrams in near-linear time. These algorithms primarily follow the same ideas as two existing algorithms to approximate optimal transport between two finite point-sets in Euclidean spaces via randomly shifted quadtrees. We show how these algorithms can be effectively adapted for the case of persistence diagrams. Our algorithms are much more efficient than previous exact and approximate algorithms, both in theory and in practice, and we demonstrate its efficiency via extensive experiments. They are conceptually simple and easy to implement, and the code is publicly available in github. 

% Indeed, persistent homology is one of the most important development in the field of topological data analysis in the past two decades. 
% Given an object, e.g, a mesh, an image, a point cloud, or a graph, by taking a specific view of how the object evolves (more formally, a filtration of it), persistent homology then maps the input potentially complex object to a topological summary, called the persistence diagram, which captures multiscale features of this objects w.r.t. this view. 
% Persistent homology thus provides a unifying way of mapping complex objects to a common feature space, the space of persistence diagrams. 

\end{abstract}

\newpage
\setcounter{page}{1}

\section{Introduction}
Recent years have witnessed a tremendous growth using topological summaries, especially the persistence diagrams (encoding the so-called persistent homology) for analyzing complex shapes. 
Indeed, persistent homology is one of the most important development in the field of topological data analysis in the past two decades \cite{EdelsLetsZomo2002, EdelsbrunnerHarer2010}. 
Given an object, e.g, a mesh, an image, a point cloud, or a graph, by taking a specific view of how the object evolves (more formally, a filtration of it), persistent homology maps the input, a potentially complex object, to a topological summary, called the persistence diagram, which captures multiscale features of this objects w.r.t. this view. 
Persistent homology thus provides a unifying way of mapping complex objects to a common feature space: the space of persistence diagrams. One can then carry out data analysis tasks of the original objects, e.g, clustering or classifying a collection of graphs, in this feature space. 
Indeed, in the past decade, persistence diagram summaries have been used for a range of applications in various domains, e.g, in material science \cite{Buchet2018,Hiraoka7035,LB17}, neuroanatomy \cite{HessNeuron,LWAPW17}, graphics \cite{CCGMO09,SOC10}, medicine /biology \cite{GambleHo2010,lfm-pro2007}, etc. 
%e.g, medicine/biology\cite{GambleHo2010}\cite{lfm-pro2007}, computer vision \cite{GuKerberGuibas2014}. 

A key component involved in such a persistent-homology based data analysis framework is to put a suitable metric on the space of persistence diagrams, and compute such distances efficiently. 
One classic distance measure developed for persistence diagrams is the $p$-th Wasserstein distance, and in practice, a popular choice for $p$ is $p = 1$, i.e, the 1-Wasserstein distance. 
This paper focuses on developing efficient, practical and light-weight algorithms to approximate the 1-Wasserstein distance for persistence diagrams. 

In particular, a persistence diagram consists of a multiset of points in the plane, where each point $(b,d)$ corresponds to the creation and death of some topological feature w.r.t. some specific filtration (view) of the input object. Given two persistence diagrams $\perD$ and $\anotherperD$, the 1-Wasserstein distance between them, denoted by $\dW(\perD, \anotherperD)$, is similar to the standard 1-Wasserstein distance (also known as the earth-mover distance) between these two multisets of planar points, but with an important distinction where points are also allowed to be matched to points in the diagonal $\aL$ (the line defined by equation $y=x$) in the plane. Intuitively, the topological features associated to points matched to the diagonal are considered as noise. 

The 1-Wasserstein distance for persistence diagrams can be computed using the Hungarian algorithm \cite{Kuhn} in $O(n^3)$ time where $n$ is the total number of points in the two persistence diagrams. This algorithm is implemented in the widely used Dionysus package \cite{Dionysus}. 
In \cite{Hera}, Kerber et al. develops a more efficient algorithm to approximate the 1-Wasserstein distance between finite persistence diagrams within constant factors. Their algorithm is based on the auction algorithm of Bertsekas \cite{BertsekasAuction}, but with a geometric twist: given that points in the persistence diagrams are all in the plane, they use the weighted kd-tree to provide more efficient search inside the auction algorithm.  In \cite{BertsekasAuction}, the time complexity of the auction algorithm is stated to be $O(A \cdot n^{1/2} \log {(n C)})$ where, in the case of persistence diagrams, $A$ is the number of possible pairings of points between persistence diagrams ($A = \Theta(n^2)$ in the worst case) and $C$ is $\max\{||x - y||_q\}$ over all possible pairings between persistence diagrams. 
While Kerber et al. did not provide an asymptotic time complexity for their approximation algorithm, they provided an empirical estimation of $O(n^{1.6})$ (not true asymptotic time complexity) by using linear regression on the observed running time versus the size of problems. They further show via extensive experiments that their approximation algorithm has a speed-up factor of 50 for small instances to a speed-up factor of 400 for larger instances in comparison to the Hungarian algorithm based implementation.

%they show via extensive experiments that their approximation algorithm has a speed-up factor of 50 for small instances to a speed-up factor of 400 for larger instances in comparison to the Hungarian algorithm based implementation. They further provided an empirical estimation of $O(n^{1.6})$ (not true asymptotic time complexity) by using linear regression on the observed running time versus the size of problems.

\subparagraph*{Related work in optimal transport for Euclidean point sets.}
As we will formally introduce in Section \ref{section:preliminaries}), 1-Wasserstein distance for persistence diagrams can be viewed as the standard 1-Wasserstein  distance for discrete planar point sets with special inclusion of points in the diagonal. In what follows, to avoid confusion, we refer to the standard 1-Wasserstein distance between point sets as the \emph{optimal transport (OT)} distance. 
Starting from \cite{Bartal96} and \cite{Bartal98}, there has been a long line of work to approximate the OT-distance for Euclidean point sets using randomly shifted quadtrees (e.g, \cite{Charikar2002,KleinbergTardos2002,IndykThaper2003,LYFC19,SatoYamada2020,Flowtree}). In particular, we consider two such approaches, the \emph{$L_1$-embedding approach} by \cite{IndykThaper2003}, and the \emph{flowtree approach} by \cite{Flowtree}. 
The former maps an input point set $P$ to a certain count-vector $V^P$ with the help of a randomly shifted quadtree, and uses the $L_1$ distance $\|V^P - V^Q\|_1$ between two such count-vectors to approximate the OT-distance between $P$ and $Q$. The latter also uses a randomly shifted quadtree and embeds input points to quadtree cells. It then shows that a certain distance computed from an optimal OT-flow induced by the tree metric (which can be computed by a greedy algorithm in linear time) can approximate the OT-distance between the original point sets. 
Let $\Delta$ denote the spread of the union of two input point sets. Both approaches give an $O(\log \Delta)$-approximation of the OT-distance between original point sets, in time $O(n\log \Delta)$. 

Recently, the idea of using metric trees to approximate OT-distance has also been extended to a more general unbalanced optimal transport problem (where $|P| \neq |Q|$) in \cite{SatoYamada2020}. In \cite{SatoYamada2020}, Sato et al. develops an $O(n\log ^2n)$ time algorithm to approximate unbalanced optimal transport on tree metrics using dynamic programming.  

\subparagraph*{New work.~} 
In practice, for applications such as nearest neighbor search, clustering and classification on large data sets, huge numbers of distance computations will be needed. The time complexity of the aforementioned algorithms for persistence diagrams using the Hungarian algorithm or the geometric variant of the Auction algorithm still causes a significant computational burden.  
In this paper, we aim to develop \emph{near-linear time} approximation algorithms for the 1-Wasserstein distance between persistence diagrams. 
Specifically: 
\begin{itemize}%\denselist 
\item In Section \ref{section:approx_alg}, we show how to modify the algorithms of \cite{IndykThaper2003} and \cite{Flowtree} to approximate the 1-Wasserstein distances between persistence diagrams within the same approximation factor (Theorems \ref{thm:mod_l1} and \ref{thm:mod_ft}). 
Note that in the literature (e.g, \cite{Hera}), it is known that $\dW(\perD, \anotherperD)$ between two persistence diagrams can be computed by (i) first augmenting $\perD$ and $\anotherperD$ to be $\widehat{\perD} = \perD \cup \pi(\anotherperD)$ and $\widehat{\anotherperD} = \anotherperD \cup \pi(\perD)$, respectively, where $\pi(x)$ projects a point $x$ to its nearest neighbor in the diagonal $\aL$; and then (ii) compute the OT-distance between $\widehat{\perD}$ and $\widehat{\anotherperD}$, although it is important to note that the cost of matching two diagonal points needs to be set to be $0$, instead of the standard Euclidean distance. 
However, this requires the modification of the cost for diagonal points; in addition, this also needs to modify a diagram $\perD$ depending on which other diagram $\anotherperD$ it is to be compared with. 
We instead develop a modification where such projection is not needed.  
\item Our modified approaches maintain the simplicity of the original approximation algorithms and are easy to implement. In comparison to approximation for unbalanced optimal transport presented in \cite{SatoYamada2020}, our modified approaches are specific to persistence diagrams and the data structures needed for both of our approaches are much simpler than those of \cite{SatoYamada2020}. Our code is publicly available in github.
In Section \ref{section:exp_res}, we present various experimental results of our new algorithms. We show that both are orders of magnitude faster than previous approaches, although at the price of worsened approximation error. However, note that in practice, the approximate factors are rather small, not as large as the worst case approximation factor. 
We also note that the modified flowtree algorithm achieves a more accurate approximation of the 1-Wasserstein distance for persistence diagrams than the modified $L_1$-embedding approach empirically, although the latter is significantly faster than the former. However, the $L_1$-embedding approach is easier to combine with proximity search data structures e.g, locality sensitive hashing (LSH), given that each input persistence diagram is mapped to a vector and the distance computation is the $L_1$-distance between two such vectors.
%We also expect that it is easier to combine the latter with proximity search data structures or strategies (such as locality sensitive hashing for approximate nearest neighbor search), given that each input persistence diagram is mapped to a vector and the distance computation is the $L_1$-distance between two such vectors.
\end{itemize}

\section{Preliminaries}
\label{section:preliminaries}
In this section, we first introduce the persistence diagrams and the 1-Wassertein distance between them, which is related to the optimal transport distance (standard 1-Wasserstein distance) for Euclidean point sets. We next describe two existing approximation algorithms for optimal transport distance \cite{Flowtree,IndykThaper2003} based on the use of randomly shifted quadtrees. Our new algorithms (in section \ref{section:approx_alg}) will be based on these two approximation algorithms. 

\subsection{Persistence Diagrams and 1-Wasserstein distance} 
\label{pd_prelim}

We first give a brief introduction of persistent homology and its associated persistence diagram summary. See \cite{EdelsbrunnerHarer2010} for a more detailed treatment of these topics. 
Suppose we are given a topological space $\XX$. A filtration of $\XX$ is a growing sequence of sub-spaces 
$$\mathbb{F}: ~~\emptyset = \XX_0 \subseteq \XX_1 \subseteq \XX_2 \subseteq \cdots \XX_m = \XX$$ 
which can be viewed as a specific way to inspect $\XX$. 
For example, a popular way to generate a filtration of $\XX$ is by taking some meaningful descriptor function $f: \XX \to \reals$ on $\XX$, and take the growing sequence of sub-level sets $\XX_a := f^{-1}(-\infty, a] = \{x \in \XX \mid f(x) \le a \}$ as $a$ increases to be the filtration.
Now given a filtration $\mathbb{F}$, through its course, new topological features (e.g, components, independent loops and voids, which are captured by the so-called homology classes) will sometimes appear and sometimes disappear. The persistent homology encodes the \emph{birth} and \emph{death} of such features in the \emph{persistence diagram} $\dgm \mathbb{F}$. In particular, $\dgm \mathbb{F}$ consists of a \emph{multiset} of points in the plane, that is, a set of points with multiplicities, where each point $(b, d)$ with multiplicity $m$ intuitively means that $m$ independent topological features (homology classes) are created in $\XX_b$ and killed in $\XX_d$. 
Thus, we also refer to $b$ and $d$ as the \emph{birth-time} and \emph{death-time}. The \emph{persistence} of this feature is $|d-b|$ which is the lifetime of this feature. We refer to points in the persistence diagram as \emph{persistent-points}. 

Note that, in general, persistent-points lie above the diagonal $\aL = \{ (x,x) \mid x\in \reals\}$ in the plane. Points closer to the diagonal $\aL$ have lower lifetime (persistence) and thus are less important, with a point $(x,x) \in \aL$ intuitively meaning a feature with persistence $0$. 

%\myparagraph{1-Wasserstein distance for persistence diagrams.} 
To compare two persistence diagrams $\perD$ and $\anotherperD$, intuitively, we wish to find a one-to-one correspondence between their multiset of points (and thus between the features they capture). However, the two sets may be of different cardinality, and we also wish to allow a persistent-point from one diagram to be ``noise" and not present in the other diagram,  which can be captured by allowing this point $p = (p.x,p.y)$ to be matched to its nearest neighbor projection $\pi(p)$ in $\aL$. Let $\pi: \reals^2 \to \aL$ be this projection, where  $\pi(p):= (\frac{p.x+p.y}{2}, \frac{p.x+p.y}{2})$. 
The following $p$-th Wasserstein distance essentially captures this intuition \cite{EdelsbrunnerHarer2010}. 
\begin{definition}[p-Wasserstein distance for persistence diagrams] \label{def:dWperD} 
Given a persistence diagram $\perD$, its \emph{augmentation} $\myaug({\perD})$ consists of $\perD$ together with all points in $\aL$ each with infinite multiplicity. 
Given two persistence diagrams $\perD$ and $\anotherperD$, with their augmentations $\myaug({\perD})$ and $\myaug(\anotherperD)$, respectively, the $p$-Wassertein distance between them is 
\begin{align}
    \mathrm{d_{W,p}^{per}}(\perD, \anotherperD) &:= \inf_{\mu: \myaug(\perD) \to \myaug(\anotherperD)} \bigg( \sum_{p\in \myaug(\perD)}  \|p - \mu(p)\|^p_q \bigg)^{1/p}, 
\end{align}
where $\mu: \myaug(\perD) \to \myaug(\anotherperD)$ ranges over all possible bijections among the two sets. 
\end{definition} 
Note that $q$ is used to denote the inner $L_p$-norm. If $p = \infty$, the $\infty$-Wasserstein distance is the classic \emph{bottleneck distance} between persistence diagrams \cite{PersistenceStability2007}\cite{Hera}. In this paper, we are interested in the case when $p = 1$. 
%While the above definition seems cumbersome, it turns out that it is  equivalent to the following version which we will use often in this paper: 
It turns out that an equivalent definition (which we will use in this paper) is as follows: 
\begin{definition}[1-Wasserstein distance for persistence diagrams, version 2] \label{def:dWperDV2}
Given two point sets $A$ and $B$ in $\reals^2$, an \emph{augmented (perfect) matching} for them is a subset $\Gamma \subset \big(A \cup \pi(B)\big) \times \big(B \cup \pi(A)\big)$ such that 
(i) each $a\in A$ or $b\in B$ appears in exactly one pair in $\Gamma$, and (ii) each $(a,b) \in \Gamma$ is of the following three forms: (1) $a\in A, b\in B$, (2) $a\in A, b = \pi(a) \in \pi(A)$, or (3) $a = \pi(b) \in \pi(B), b\in B$. 

Given two persistence diagrams $\perD$ and $\anotherperD$, the 1-Wasserstein distance between them is: 
\begin{align} 
\dW(\perD, \anotherperD) &:= \min_{\Gamma} \sum_{(p,q)\in \Gamma} \|p - q\|_p, 
\end{align}
where $\Gamma$ ranges over all possible augmented matchings for $\perD$ and $\anotherperD$. 
\end{definition}

\subsection{Relation to optimal transport} 
Readers may have already noticed the similarity between Definition \ref{def:dWperD} with the standard $p$-th Wasserstein distance between two  probability measures. 
To avoid confusion, 
%we now focus on the case when the two input measures $\mu, \nu: \reals^d \to \reals$ are defined over $\reals^d$ and have discrete support \footnote{Given a probability $\mu: \reals^d\to \reals$, its support $\mysupp(\mu):= \{x\in \reals^d \mid \mu(x) > 0\}$.}. We refer to the 
from now on we refer to 1-Wassertein distance as \emph{optimal transport} so as to differentiate from the use of 1-Wasserstein distance of persistence diagrams. 
\begin{definition}[Optimal transport] \label{def:OT}
Given a finite metric space $(X, d_X)$ and two measures $\mu, \nu\in X \to \reals$, the \emph{optimal transport} between them is 
\begin{align}
    \dOT(\mu, \nu) &:= \min_{\tau: X \times X \to \reals} \sum_{x, y \in X} \tau(x,y) \cdot d_X(x,y), 
\end{align}
where $\tau$, called a \emph{transport plan} or \emph{a flow}, is a measure on $X\times X$ whose marginals equal to $\mu$ and $\nu$, respectively; that is, $\tau(\cdot, Y) = \mu(\cdot)$ and $\tau(X, \cdot) = \nu(\cdot)$. 
\end{definition}

Given a multiset of points $A$ in the plane, note that we can view this as a discrete measure 
supported on points in $A$, such that for each subset $S$ of $A$, $\mu_A(S) = \sum_{a \in S} c_a\delta_a$ where $c_a$ is the multiplicity of $a$ in $A$, 
%where $c$ is the ratio between the multiplicity of $a$ over the total cardinality of points in $A$, 
while $\delta_a$ is the Dirac measure supported at $a$. 
Hence in what follows, we sometimes abuse the notations and equate a multiset of points with the discrete measure induced by it, and talk about optimal transport between two multisets of points. 

As shown in \cite{Hera}, one can consider $\widehat{\perD}:=\perD \cup \pi(\anotherperD)$ and $\widehat{\anotherperD}:=\anotherperD \cup \pi(\perD)$ and modify the Euclidean distance so that $d(x, y) = 0$ for $x, y \in \pi(P) \cup \pi(Q)$ to obtain a modified pseudo-metric space $(\reals^2, d)$. In this case, $\dW(\perD, \anotherperD)$ becomes the optimal transport between the discrete measures induced by $\widehat{\perD}$ and $\widehat{\anotherperD}$ under this modified pseudo-metric.

We can also relate $\dW(\perD, \anotherperD)$ to the optimal transport between the discrete measures induced by $\widehat{\perD}$ and $\widehat{\anotherperD}$ 
with the following observation (simple proof is in Appendix \ref{appendix:proofs}): 
\begin{observation}
\label{ot_w1}
Let $\mu_{\widehat{\perD}}$ be the discrete measure induced by $\widehat \perD$ and $\nu_{\widehat{\anotherperD}}$ be the discrete measure induced by $\widehat \anotherperD$. Then
$\dOT(\mu_{\widehat{\perD}}, \nu_{\widehat{\anotherperD}}) \leq 2 \cdot \dW(\perD, \anotherperD).$
\end{observation}
Given two discrete measures $\mu, \nu \in X \times \reals$ on a finite metric space $(X, d_X)$, computing the optimal transport distance can be reduced to finding the optimal min-cost flow on a complete bipartite graph using combinatorial flow algorithms as described in \cite{Kuhn}. 
In our setting later, $\mu$ and $\nu$ will both be induced by point sets in $\reals^2$, and $d_X$ is the standard Euclidean distance.

\subsection{Quadtree-based approximation algorithms for optimal transport}
\label{subsection:qt_prelim}

In this section, we briefly review two algorithms to approximate the optimal transport for two discrete measures $\mu$ and $\nu$. 
%While this part is developed for L_2 ground metric only, it can be extended to any ground metric and only changes a constant factor in the distortion. 
Both of these algorithms use a randomly shifted quadtree, which we introduce first. 

\subparagraph*{Randomly-shifted quadtree.} 
Let $X \subseteq \R^d$ be a finite set of points (for our setting, $d=2$ for persistence diagrams). To simplify the description, we will assume that the minimum pairwise distance between any two points in $X$ is 1 and that $X$ is contained in $[0, \Delta]^d$ (where $\Delta$,  the ratio of the diameter of $X$ over the minimum pairwise distance, is also called the \emph{spread} of $X$). 
First, let $H_0 = [-\Delta, \Delta]^d$ be the hypercube with side length $2\Delta$ which is centered at the origin. 
Now shift $H_0$ by a random vector whose coordinates are from $[0, \Delta]$ to obtain $H$. Note that $H$ still encloses $X$ as $H$ has side length $2\Delta$. 
%Randomly shift $H_0$ along each dimension by defining $\sigma \in \R^d$ to be a random vector which coordinates from $[0, \Delta]$ and shifting $H_0$ by $\sigma$ to get $H$. Note that $H$ encloses $X$ and has side length $2\Delta$. 

Construct a tree of hypercubes by letting $H$ be associated to the root and halving $H$ along each dimension. Recurse on the resulting sub-hypercubes that contain at least one point from $X$, and stop when a hypercube contains exactly one point from $X$. 
Each leaf node of resulting quadtree $\aT_X$ contains exactly one point in $X$, and there are exactly $|X|$ leaves.
The resulting quadtree $\aT_X$ has at most $O(\log (d\Delta))$ levels. To see that $\aT_X$ has at most $O(\log(d\Delta))$ levels, consider the depth $i$ of some internal node. We know that the hypercube associated with the node has a side length of $\frac{\Delta}{2^i}$ and the distance between any two points in the hypercube, $c$, is less than or equal to $\frac{\Delta \sqrt{d}}{2^i}$. Then $i \leq \log (d \Delta)$ so there are at most $O(\log (d\Delta))$ levels in $\aT_X$. Additionally, the size of $\aT_X$ is $O(|X| \log (d\Delta))$.
It can be constructed in $\tilde{O}(|X| \log (d\Delta))$ time where $\tilde{O}$ includes term polynomial in $\log |X|$. 
We set the root level as level $\log \Delta + 1$ and subsequent levels are labeled as $\log \Delta, \log \Delta - 1 , \dots$. The weight of each tree edge between level $\ell + 1$ and level $\ell$ is $2^\ell$. Note that the quadtree cell has side length $2^{\ell}$ at level $\ell$.

\subparagraph*{Approximation algorithm 1: $L_1$-embedding via $\aT_X$.} 
Given two discrete measures $\mu$ and $\nu$, let $X$ be the union of their support \footnote{Note that in general, $X$ can be a superset of the support of $\mu$ and $\nu$. Indeed, if there are a set of $m$ measures and we perform kNN queries for a query measure, it is more convenient to set $X$ as the union of support of all these measures and build only a single quadtree $\aT_X$.}. Construct the randomly shifted quadtree $\aT_X$ as described above; $X$ is sometimes omitted from the subscript when its choice is clear from the context. 
Given a tree node $v\in \aT_X$, its level is denoted by $\ell(v)$. We will abuse the notation slightly and use $v$ also to denote the quadtree cell (which is a hypercube of size length $2^{\ell(v)}$). Given a discrete $\mu$, then $\mu(v)$ denotes the total measure of points from $\mu$ contained within this quadtree cell, namely, the total size of points with multiplicity counted from $\mu$ within this quadtree cell.
We can now map $\mu$ to a vector $\aV^\mu$ where each index corresponds to a tree node $v \in \aT_X$, and $\aV^\mu[v]$ has coordinates $2^{\ell(v)} \mu(v)$. Similarly, map $\nu$ to vector $\aV^\nu$. 
Then Indyk and Thaper \cite{IndykThaper2003} showed that $\| \aV^\mu - \aV^\nu\|_1 = \sum_{v\in \aT} 2^{\ell(v)} |\mu(u) - \nu(v)|$ gives an approximation to the optimal transport $\dOT(\mu, \nu)$ in expectation. 

\begin{theorem}[\cite{IndykThaper2003}] \label{thm:L1embedding} 
Given two discrete measures $\mu, \nu$ such that $\mysupp(\mu) \cup \mysupp(\nu) \subseteq \R^2$ and $s = |\mysupp(\mu) \cup \mysupp(\nu)|$, using a randomly shifted quadtree, $\|\aV^\mu- \aV^\nu\|_1$ can be calculated in time $O(s\log \Delta)$ and there are constants $C_1, C_2$ such that $C_1 \cdot \dOT(\mu, \nu) \leq E[\|\aV^\mu - \aV^\nu\|_1] \leq C_2 \cdot \log \Delta \dOT(\mu, \nu)$. Here, $E[ \cdot]$ stands for the expectation.
\end{theorem} 

We note that it also turns out that $\| \aV^\mu - \aV^\nu\|_1$ gives exactly the optimal transport between $\mu$ and $\nu$ along the tree metric induced by $\aT_X$. Specifically, for each $v\in \aT_X$, set its weight to be $w(v) = 2^{\ell(v)}$.
Then for any $x, x' \in X$, define $d_T(x, x')$ to be the total weight of the unique tree path connecting the quadtree leaf $v_x$ (containing $x$) and leaf $v_{x'}$ (containing $x'$).
Then the optimal transport between $\mu$ and $\nu$ w.r.t. metric $d_T$, denoted by $\dOTtree(\mu, \nu)$,
satisfies that $\dOTtree(\mu, \nu) = \| \aV^\mu - \aV^\nu\|_1$. 

Furthermore, we can consider that this optimal transport $\dOTtree$ is generated by a following \emph{greedy-flow} $\Gflow: X \times X \to \reals$: 
Starting from leaf-nodes, we will match up as many unmatched points $\mu\cap v$ to $\nu \cap v$ as we can within each node $v$, and pass the remaining unmatched portion to its parent. In general, each tree node $v\in \aT$ will have a $\mu$-demand $\widehat{\mu}(v)$ and $\widehat{\nu}(v)$, which collect all unmatched measure from its $2^d$ child nodes. $\mu$-demand (resp. $\nu$-demand) at a leaf node  $v$ is initialized to be $\mu(v)$ (resp. $\nu(v)$). We then match these demand as much as we can and pass on $|\widehat{\mu}(v) - \widehat{\nu}(v)|$ to its parent as unmatched $\mu$-measure, or unmatched $\nu$-measure, whichever is left. 
Note that a greedy-flow $f_G$ is not unique, but it tuns out that any such greedy-flow (greedy transport plan) gives rise to the optimal transport distance between $\mu$ and $\nu$ w.r.t. the tree metric $d_T$ (See \cite{Kalantari95} for more detail): i.e, 
\begin{align}\label{eqn:greedyflow} 
    (\| \aV^\mu - \aV^\nu\|_1 = )~ \dOTtree(\mu, \nu) &= \sum_{x, x'\in X} \Gflow (x, x') d_T(x, x'). 
\end{align}

\subparagraph*{Approximation algorithm 2: Flowtree.} 
The flowtree algorithm by \cite{Flowtree} is based on the previous approach. The only modification is that, consider a greedy-flow $\Gflow$ as described above. 
Instead of using the tree metric $d_T$ to compute the optimal transport distance, the \emph{flowtree estimate} computes the cost of this flow using the standard Euclidean distance:  
%That is, we define the following \emph{flowtree estimate}
%Instead of using the tree metric $d_T$ to compute the optimal transport distance, now, compute the cost of this flow using the standard Euclidean distance. 
%That is, we define the following \emph{flowtree estimate} 
\begin{align}\label{eqn:flowtreeW}
\dOTflow (\mu, \nu) &= \sum_{x, x' \in X} \Gflow (x, x') \| x - x'\|. 
\end{align}
Comparing Equation (\ref{eqn:greedyflow}) to the above equation, the difference is minor ($d_T(x,x')$ versus $\|x - x'\|$). However, in practice, $\dOTflow$ appears to provide a much more accurate estimate to the optimal transport distance $\dOT(\mu, \nu)$ w.r.t. the Euclidean distance. 
Unfortunately, unlike $\dOTtree$, which can be computed as a $L_1$-distance between two specific vectors, to compute $\dOTflow$, we now have to compute a greedy-flow $\Gflow$ explicitly (which can be done linear in the size of quadtree; however conceptually, this is not as simple as $L_1$-distance). Overall, we have the following result: 
\begin{theorem}[\cite{Flowtree}]
Given two discrete measures $\mu, \nu$ such that $\mysupp(\mu) \cup \mysupp(\nu) \subseteq \R^2$ and $s = |\mysupp(\mu) \cup \mysupp(\nu)|$, using a randomly shifted quadtree, $\dOTflow(\mu, \nu)$ can be computed in time $O(s \log \Delta)$ and there are constants $C_1, C_2$ such that $C_1 \cdot \dOT(\mu, \nu) \leq E[\dOTflow(\mu, \nu)] \leq C_2 \cdot \log \Delta \cdot \dOT(\mu, \nu)$. Here, $E[\cdot ]$ stands for the expectation. 
\end{theorem}

\section{Approximating 1-Wasserstein distances for persistence diagrams} 
\label{section:approx_alg}

We now present two algorithms to approximate the 1-Wasserstein distance for persistence diagrams, based on the approximation schemes of optimal transport in Section \ref{subsection:qt_prelim}. Note that the results here are developed for the $L_2$ norm and through the equivalence of norms, can be generalized to any $L_p$ norm and only changes the constant factor in the distortion induced by each approximation. 
%. Our approaches are based on the $L_1$-embedding for optimal transport \cite{IndykThaper2003} and the flowtree \cite{Flowtree} algorithms (both presented in Section \ref{subsection:qt_prelim}), respectively. 
%We first describe an $L_1$ embedding for persistence diagrams similar to that of \cite{IndykThaper2003}. Note that the $L_1$ embedding described in \cite{IndykThaper2003} is the same as using the quadtree approximation as described in section \ref{subsection:qt_prelim}. We then modify the flowtree approach described in \cite{Flowtree}.

\subsection{Approximation algorithms via $L_1$ embedding}
\label{subsec:PerL1embedding} 

%\noindent{\bf Description of algorithm.~}   
\subsubsection{Description of the new quadtree-based $L_1$-embedding}
\label{subsubsec:modl1} 
Let $\perD$ and $\anotherperD$ be two persistence diagrams and let $X = \perD \uplus \anotherperD$, the disjoint union of $\perD$ and $\anotherperD$. In what follows, for simplicity of presentation, we assume that the minimum distance between any two distinct points in $X$, as well as between any point in $X$ with a point in the diagonal $\aL$, is $1$. The latter constraint can be removed with some extra care on handling leaf nodes in the quadtree. Assume w.l.o.g that $\Delta$ is a power of $2$. 

Partition the (randomly shifted) hypercube $H$ described in section \ref{subsection:qt_prelim} into grids where the cells have side length $\Delta, \Delta/2, \ldots, 2^i, \ldots, 2, 1, \frac{1}{2}$. Note that each cell at the lowest level can contain at most one point, and if a leaf contains a point then it cannot intersect the diagonal $\aL$. 
Let $\aT_X$ be the resulting quadtree, where leaves are all cells that contain exactly one point from $X$. 
We further use $G_i$ to denote the set of quadtree cells with side-length $2^i$ (i.e, those in level-$i$); we refer to $G_i$ as the level-$i$ grid. Note that the size of the quadtree is $O(|X| \log \Delta)$. Additionally, we call a cell a \emph{terminal cell} if it intersects the diagonal $\aL$; otherwise, it is \emph{non-terminal}. 

Now for each grid $G_i$, construct a vector $\aV_i^\perD$ with one coordinate per cell, where each coordinate counts the number of points in the corresponding cell. The vector representation $\aV^\perD$ for $\perD$ is then the concatenation of all these vectors $2^i\aV_i^\perD$ where $2^i$ is the cell side length for grid $G_i$:  
    \begin{equation*}
     \aV^\perD = [ \frac{1}{2}\aV_{-1}^\perD, \aV_0^\perD, 2\aV_1^\perD, \dots, 2^i\aV_i^\perD, \dots, ] 
    \end{equation*} 
Construct the vector $\aV^\anotherperD$ similarly.  
We use $p_k$ to denote the value of coordinate $k$ in $\aV_i^\perD$ and $q_k$ to denote the value of coordinate $k$ in $\aV_i^\anotherperD$. 
    Now, we will describe a \emph{modified-$L_1$ distance} $|\aV^\perD - \aV^\anotherperD |_T$ for these vectors, 
    which is similar to the $L_1$ norm. 
    To compute $|\aV_i^\perD - \aV_i^\anotherperD|_T$, we will define $|p_k - q_k|_T$. There are two cases for the $|p_k - q_k|_T$ to consider: \\
    Case 1: if coordinate $k$ is not associated with a terminal cell, then use $|p_k - q_k|$ for $|p_k - q_k|_T$.\\
    Case 2: if coordinate $k$ is associated with a terminal cell, then set $|p_k - q_k|_T = 0$. \\ 
    Then we have $|\aV_i^\perD - \aV_i^\anotherperD|_T = \sum_{k = 1}^{|G_i|}|p_k - q_k|_T$, and  
    \begin{align}
     %   |\aV_i^\perD - \aV_i^\anotherperD|_T & = \sum_{k = 1}^{|G_i|}|p_k - q_k|_T, ~~\text{and} \nonumber \\
%    \end{equation}
%    where $n$ is the dimension of $\|\aV_i^\perD - \aV_i^\anotherperD\|_T$ and let 
        \dTper(\perD, \anotherperD) &:=|\aV^\perD - \aV^\anotherperD|_T  = \sum_{i = -1}^{\log_2\Delta} 2^i |\aV_i^\perD - \aV_i^\anotherperD|_T. 
        \label{eqn:dTper} 
    \end{align}

    \subparagraph*{An equivalent $L_1$-distance formulation.} 
    We introduce the above vector representation and the modified $L_1$-distance as it is more convenient for later theoretical analysis. However, algorithmically, we wish to have a true $L_1$-embedding. It turns out that an equivalent formulation is as follows: Let $\widehat{G}_i$ denote the level-$i$ quadtree cells \emph{that do not intersect the diagonal $\aL$}. We then compute a vector representation $\widehat{\aV}^\perD$ (resp. $\widehat{\aV}^\anotherperD$) restricted only to cells in $\bigcup \widehat{G}_i$. In other words, all entries corresponding to cells intersecting the diagonal are ignored in constructing $\widehat{\aV}^\perD$  and $\widehat{\aV}^\anotherperD$. We then have that 
    \begin{align}
        \|\widehat{\aV}^\perD - \widehat{\aV}^\anotherperD \|_1 = |\aV^\perD - \aV^\anotherperD|_T  = \dTper(\perD, \anotherperD). 
    \end{align}
    That is, our quadtree-induced distance $\dTper(\perD, \anotherperD)$ is a $L_1$-distance for suitably constructed vectors. Nevertheless, we use the definition as in Equation (\ref{eqn:dTper}) to simplify proofs later. 
    
It is easy to see that the construction takes the same time as the $L_1$-embedding approach described in Section \ref{subsection:qt_prelim}. We now show that the $L_1$-distance $\dTper(\perD,\anotherperD)$ approximates the 1-Wasserstein distance $\dW(\perD, \anotherperD)$ for the persistence diagrams. The main results for this $L_1$-embedding approach are summarized as follows, and we prove the approximation bound in Section \ref{subsubsec:approxL1}. 
 \begin{theorem}
    \label{thm:mod_l1}
    Given persistence diagrams $\perD$ and $\anotherperD$ such that $s = |\perD|+ |\anotherperD|$, we can compute $\dTper(\perD, \anotherperD) = |\aV^\perD - \aV^\anotherperD|_T$ in time $O(s \log \Delta)$ using the randomly shifted quadtree. 
    Furthermore, the expected value of $\dTper(\perD, \anotherperD)$ is an $O(\log \Delta)$-approximation of the 1-Wasserstein distance $\dW( \perD, \anotherperD)$; i.e, 
    there are constants $c_1$ and $c_2$ such that $c_1 \cdot \dW(\perD, \anotherperD) \leq E[|\aV^\perD - \aV^\anotherperD|_T] \leq c_2 \log \Delta \cdot \dW( \perD, \anotherperD)$.
    \end{theorem}
    
\subsubsection{Approximation guarantees} 
\label{subsubsec:approxL1}

Our approximation bound in Theorem \ref{thm:mod_l1} follows from Lemmas \ref{lemma:lower_bound_l1} and \ref{lemma:upper_bound_l1} below. To prove these lemmas, we will first introduce a \emph{greedy augmented matching}. 

\subparagraph*{Greedy augmented matching $\greedyM$.} We construct the following augmented matching (recall Definition \ref{def:dWperDV2}) $\greedyM \subseteq (\perD \cup \pi(\anotherperD)) \times (\anotherperD \cup \pi(\perD)\big)$ in a \emph{bottom-up greedy manner}: 
Starting from the level $i = -1$, we will aim to match points in $\perD \uplus \anotherperD$ as much as we can within each level $G_i$. Those remaining unmatched points will then be considered at the next level $G_{i+1}$: in particular, within each \emph{non-terminal cell} in $G_{i+1}$, we will match the maximal possible unmatched points in $\perD$ to unmatched points in $\anotherperD$ so far, and pass the remainder unmatched points (which can now only come from either $\perD$ or $\anotherperD$, but not both) to its parents. 
Within a \emph{terminal cell $v$} in $G_{i+1}$, we match every unmatched point $p$ from $\perD \cap v$ and from $\anotherperD \cap v$ to its closest point $\pi(p) \in \aL$ in the diagonal.  
Finally, at the root cell (in level $\log \Delta$), any unpaired points from either $\perD$ or $\anotherperD$ will be paired to the diagonal, as the root is a terminal cell. 
Note that by construction, at any level $i$, $|\aV_i^\perD - \aV_i^\anotherperD|_T$ is exactly the number of points from $\perD \uplus \anotherperD$ that could not be matched at level $i$ or below under such a greedy augmented matching and will subsequently need to be matched in grid $G_j$, $j \ge i+1$. 
See Figure \ref{fig:example} for a simple illustration. 
%\yusu{Sam: I suggest we have a figure containing a simple example, to both show quadtree, vector representations, and also the greedy augmented matching.} \sam{Figure \ref{fig:example}}
\begin{figure}
    \centering
    \includegraphics[scale=0.2]{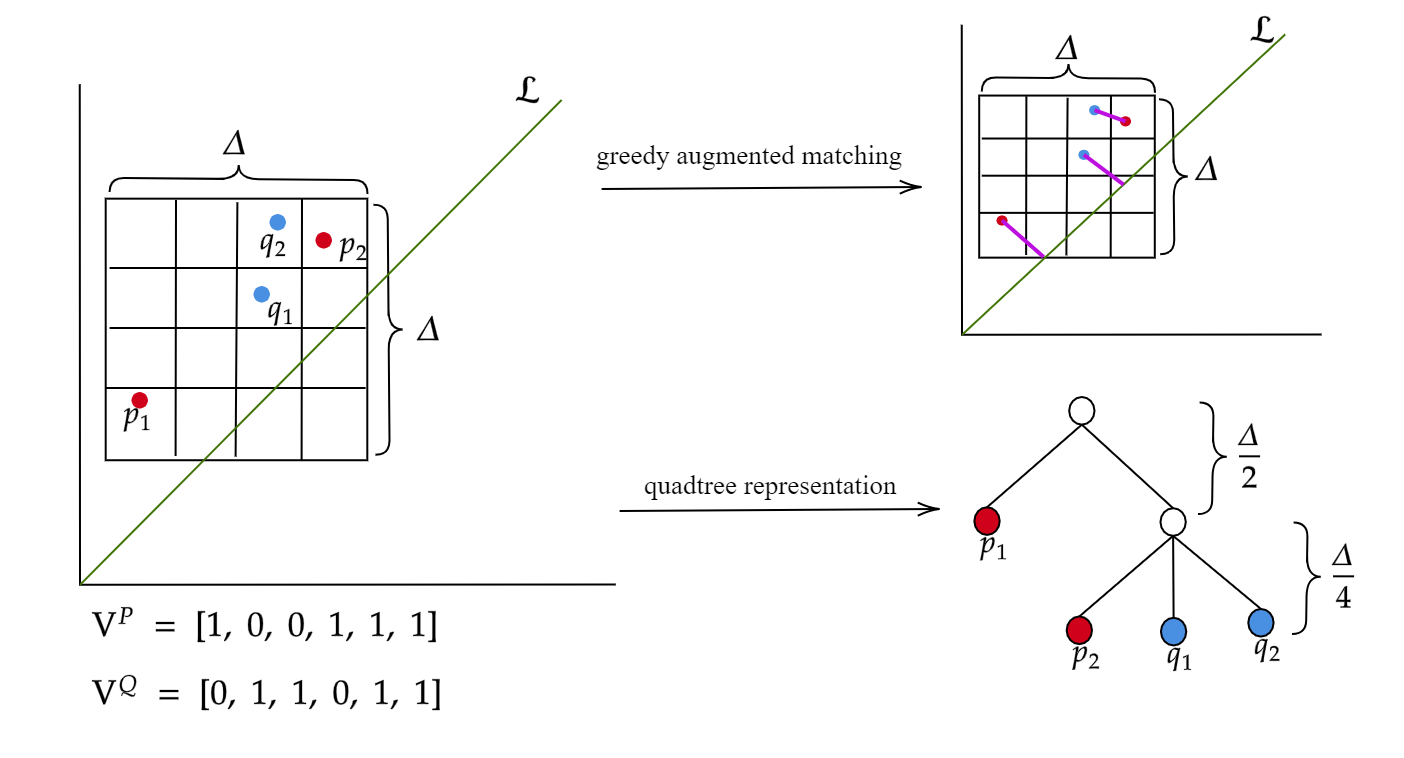}
    \vspace*{-0.15in}
    \caption{A simple example of the vector and quadtree representations of the persistence diagrams $P = \{p_1, p_2\}$ and $Q = \{q_1, q_2\}$.}
    \label{fig:example}
\end{figure}

\begin{lemma}
    There is a constant $C$ such that $\dW( \perD, \anotherperD) \leq C \cdot |\aV^\perD - \aV^\anotherperD|_T$.
    \label{lemma:lower_bound_l1}
\end{lemma}
\begin{proof}    
Consider the greedy augmented matching $\greedyM$ we described above, induced by pairing points or pairing points with their diagonal projection greedily within the same cells of the grids $G_{-1}, G_0, G_1, \dots$ in a bottom-up manner. 
First, given $(p,q) \in \greedyM$, we say that $(p,q)$ is paired in level-$i$, if the lowest level any quadtree cell containing both $p, q$ is level-$i$; intuitively, $(p,q)$ are paired greedily in this cell in level-$i$. 
We now use $\greedyM_i \subseteq \greedyM$ to denote the set of pairs from level-$i$. 
For each level $i$, let $cost(i) = \sum_{(p, q) \in \greedyM_i} \|p - q\|_2$ be the total cost incurred by all those pairs from $\greedyM_i$, and obviously,
\begin{align}\label{eqn:greedycost} 
    \sum_i cost(i) &= \sum_{(p,q)\in \greedyM} \|p - q\|_2 \ge \dW(\perD, \anotherperD), 
\end{align} 
where the right inequality holds as $\dW(\perD, \anotherperD)$ is the smallest total cost of any augmented matching and the greedy augmented matching $\greedyM$ is an augmented matching. 

There are no pairings induced in the grid $G_{-1}$ (of size $1/2$) since the minimum inter-point distance is $1$ and the minimum distance between a point and its diagonal is $1$ as well. Hence we have that $|\aV_{-1}^\perD - \aV_{-1}^\anotherperD|_T = |\perD| + |\anotherperD|$. 
Since all points from $P$ and $Q$ will remain unpaired in level $-1$ (i.e, within cells of grid $G_{-1}$), we then have that there are exactly 
$$|\perD| + |\anotherperD| - |\aV_0^\perD - \aV_0^\anotherperD|_T = |\aV_{-1}^\perD - \aV_{-1}^\anotherperD|_T - |\aV_0^\perD - \aV_0^\anotherperD|_T $$
total number of points from $\perD \uplus \anotherperD$ that can either be paired to each other or matched to the diagonal in grid cells $G_0$. 
As the maximal distance between any pair of matched points is $2^i\sqrt{2}$ within a cell in $G_i$, the maximum cost incurred by all matched points in $G_0$ is 
$cost(0) \le \sqrt{2} (|\aV_{-1}^\perD - \aV_{-1}^\anotherperD|_T - |\aV_0^\perD - \aV_0^\anotherperD|_T).$ 
In general, the maximum cost of the matched points in $G_i$ is 
$$cost(i) \le 2^i \cdot \sqrt{2} (|\aV_{i - 1}^\perD - \aV_{i - 1}^\anotherperD|_T - |\aV_i^\perD - \aV_i^\anotherperD|_T).$$ 
Hence combining Equation (\ref{eqn:greedycost}), the total cost (i.e, $\sum_{(p, q) \in \greedyM} ||p - q ||_2$) of the greedy augmented matching $\greedyM$, is bounded from above by: 
\begin{equation*}
    \sum_{(p, q) \in \greedyM} ||p - q ||_2 \leq \sum_{i = 0}^{\log_2 \Delta + 1} 2^i \cdot \sqrt{2} (|\aV_{i - 1}^\perD - \aV_{i - 1}^\anotherperD|_T - |\aV_i^\perD - \aV_i^\anotherperD|_T) \le 2\sqrt{2} |\aV^P - \aV^Q|_T. 
\end{equation*}
By the right inequality of Equation (\ref{eqn:greedycost}), the claim then follows. 
\end{proof}

    \begin{lemma}
    There is a constant $C'$ such that the expected value of $\dTper(\perD, \anotherperD)$ is bounded by
        $E[\dTper(\perD, \anotherperD)] = E[|\aV^\perD - \aV^\anotherperD|_T] \leq C' \cdot \log \Delta \cdot \dW( \perD,  \anotherperD)$. 
    \label{lemma:upper_bound_l1}
    \end{lemma}
\begin{proof} 
Set $\widehat \perD = \perD \cup \pi (\anotherperD)$ and $\widehat \anotherperD = \anotherperD \cup \pi (\perD)$ as before. 
    For a given grid $G_i$ and some coordinate $k$ in $\aV_i^\perD$ and $\aV_i^\anotherperD$, let $p_k$ be the value of coordinate $k$ in $\aV_i^\perD$ and $q_k$ be the value of coordinate $k$ in $\aV_i^\anotherperD$. 
Analogously, let $\aV^{\widehat{\perD}}$ and $\aV^{\widehat{\anotherperD}}$ be the vector representations w.r.t multisets $\widehat{P}$ and $\widehat{Q}$, respectively, and let $\widehat p_k$ (resp. $\widehat q_k$) be the value of coordinate $k$ in $\aV_i^{\widehat \perD}$ (resp. $\aV_i^{\widehat \anotherperD}$). 
    Note that $\widehat p_k \ge p_k$ and $\widehat q_k\ge q_k$. 
    
    (i) Now if $\widehat p_k > p_k$ or $\widehat q_k > q_k$, then there exists at least one point $x \in \pi(\perD) \cup \pi(\anotherperD)$ in the cell $v$ associated with coordinate $k$. In other words, this cell $v$ must be a terminal cell, and $|p_k - q_k|_T = |\widehat p_k - \widehat q_k|_T = 0$. 
    
    (ii) Otherwise if the conditions in (i) do not hold, it must be that $\widehat{p_k} = p_k$ and $\widehat{q_k} = q_k$, in which case we also have that $|p_k - q_k|_T = |\widehat p_k - \widehat q_k|_T$. 
    
    Combining (i) and (ii) we then have that $|\aV^\perD -\aV^\anotherperD|_T \leq |\aV^{\widehat \perD} - \aV^{\widehat\anotherperD}|_T$. 
    
    On the other hand, by the definition of metric $| \cdot |_T$, we know that $|\widehat p_k - \widehat q_k|_T \leq |\widehat p_k - \widehat q_k|$, implying that  $|\aV^{\widehat\perD} - \aV^{\widehat\anotherperD}|_T \leq \|\aV^{\widehat\perD} - \aV^{\widehat\anotherperD}\|_1$. 
    Let $\mu_{\widehat\perD}$ and $\nu_{\widehat\anotherperD}$ be the discrete measures induced by $\widehat{\perD}$ and $\widehat{\anotherperD}$ respectively. 
    By Theorem \ref{thm:L1embedding}, there is some constant $C$ such that $E[\|\aV^{\widehat\perD} - \aV^{\widehat\anotherperD}\|_1] \leq C \cdot \log \Delta \cdot \dOT(\mu_{\widehat \perD}, \nu_{\widehat \anotherperD})$. From Observation \ref{ot_w1}, $\dOT(\mu_{\widehat \perD}, \nu_{\widehat \anotherperD}) \leq 2 \cdot \dW(\perD, \anotherperD)$. Therefore, 
    \begin{equation*}
        E[|\aV^{\perD} - \aV^{\anotherperD}|_T] \leq E[|\aV^{\widehat\perD} - \aV^{\widehat\anotherperD}|_T]\leq E[\|\aV^{\widehat\perD} - \aV^{\widehat\anotherperD}\|_1] \leq 2 \cdot C \cdot \log \Delta \dW( \perD, \anotherperD). 
    \end{equation*}
    The lemma then follows. 
\end{proof}

\subsection{Approximation algorithm via flowtree}

We now propose an alternative approximation algorithm for $\dW(\perD, \anotherperD)$. The high level idea is the same as the flowtree algorithm described in Section \ref{subsection:qt_prelim}: in particular, we first compute the optimal flow for points in $\perD$ and $\anotherperD$ along the randomly shifted quadtree $\aT$ as constructed earlier, but now with the modification that a point can be paired to diagonal. It turns out that this leads to the same greedy augmented matching $\greedyM$ we described at the beginning of Section \ref{subsubsec:approxL1}. 
%in the proof of Lemma  \ref{lemma:lower_bound_l1} in Section \ref{subsec:PerL1embedding}. 
Then, similar to flowtree, under this greedy augmented matching $\greedyM$, we use the Euclidean distance between a pair of matched points (instead of using the tree-distance as for $\dTper(\perD, \anotherperD)$) to measure the cost of each pair of matched points. This leads to the following 
\textit{modified flowtree estimate}: 
\begin{equation}
    \dW_F(\perD, \anotherperD) = \sum_{(p, q) \in \greedyM} ||p - q||_2, 
\end{equation}
and in our second algorithm, we will use $\dW_F(\perD, \anotherperD)$ as an approximation of the true 1-Wasserstein distance $\dW(\perD, \anotherperD)$. 

From an implementation point of view, unlike $\dTper(\perD, \anotherperD)$, which can be computed as the $L_1$-distance between two vectors, we now must explicitly compute the greedy augmented matching, $\greedyM$ between $\perD$ and $\anotherperD$. (Note that this greedy augmented matching was only used in proving the approximation guarantee for $\dTper(\perD, \anotherperD)$, and not needed for its computation.) 
Computing this greedy augmented matching (and calculating its cost) takes the same time as computing the greedy flow in the original flowtree algorithm  \ref{subsection:qt_prelim}. 
Furthermore, it is easy to see that in the proof of Lemma \ref{lemma:lower_bound_l1}, we in fact showed that $\dW(\perD, \anotherperD) \le \dW_F(\perD, \anotherperD) \le C \cdot |\aV^P - \aV^Q|_T$ (see Eqn (\ref{eqn:greedycost})). 
Combining this with Theorem \ref{thm:L1embedding}, we thus obtain the following approximation result for this modified flowtree estimate. 

\begin{theorem}
\label{thm:mod_ft}
Given two persistence diagrams $\perD$ and $\anotherperD$ where $s = \max(|P|, |Q|)$ and $\Delta$ is the spread of point set $\perD \cup \anotherperD$, we can compute $\dW_F(\perD, \anotherperD)$ in time $O(s\log\Delta)$ using a randomly shifted quadtree.  
Additionally, the expected value of $\dW_F(\perD, \anotherperD)$ is an $O(\log \Delta)$-approximation of the 1-Wasserstein distance $\dW(\perD, \anotherperD)$; i.e. there are constants $C_1$ and $C_2$ such that $C_1 \cdot \dW(\perD, \anotherperD) \leq E[\dW_F(\perD, \anotherperD)] \leq C_2 \cdot \log \Delta \cdot \dW(\perD, \anotherperD)$.  
\end{theorem}

%\paragraph{Remark.} 
\noindent{{\bf Remark.}} We remark that while these two approximation schemes, $\dTper(\perD, \anotherperD)$ and $\dW_F(\perD, \anotherperD)$, have similar approximation guarantees for $\dW(\perD, \anotherperD)$, in practice, the modified flowtree based approach has much higher accuracy. This is consistent with the performance of flowtree algorithm versus the $L_1$-embedding approach for general optimal distance \cite{Flowtree}. In contrast, the benefit of the $L_1$ embedding approach is that it is easy to compute  $\dTper(\perD, \anotherperD)$. Also, each persistence diagram is now mapped to a vector representation, and the distance is $L_1$-distance among these vector representations. One could combine this with methods such as locality-sensitive-hashing for more efficient approximate $k$-nearest neighbor queries. In general, such a $L_1$-norm also makes $\dTper(\cdot, \cdot)$ potentially more suitable for downstream machine learning pipelines.

\section{Experimental results}
\label{section:exp_res}
We evaluate both the runtime and accuracy of the modified flowtree and $L_1$-embedding against the \textit{Hera} method of \cite{Hera}. We use the implementation of Hera provided in the GUDHI library \cite{gudhi} for testing. We run the experiments using an Intel Core i7-1065G7 CPU @ 1.30 GHz and 12.0GB RAM. Additionally, the implementations for both the modified flowtree and $L_1$-embedding were done in C++ (wrapped in python for evaluation) and are based on the code provided in \cite{Flowtree}. 

\subparagraph*{Datasets.} For our experiments, we use both synthetic persistence diagrams, as well as persistence diagrams generated from real data. For synthetic data, we generate two sets of persistence diagrams: called ``\emph{synthetic-uniform}" and  ``\emph{synthetic-Gaussian}", which are generated by a uniform sample and a sample w.r.t. a Gaussian distribution on the birth-death plane to obtain the persistence diagrams, respectively. For real datasets, we use persistence diagrams generated from the so-called Reddit data sets (which is a collection of graphs) \cite{real_datasets}, and from the ModelNet10 \cite{ModelNet} dataset of shapes. Details of these datasets are in Appendix \ref{appendix:datasets}. 

\subparagraph*{Speed comparison.} We compare the running time of our new approximation algorithms with that of Hera \cite{Hera} -- note that we do not directly compare with the exact algorithm by Dionysus, because as reported by \cite{Hera}, Hera is 50 times to 400 times faster than Dionysus. 
Note that Hera is also an approximation algorithm, and there is a parameter $\eps$ to adjust its approximation factor $(1 +\eps)$. By setting this parameter to be very small ($\eps = 0.01$), we use the distance computed by Hera as ground truth later when we measure approximation accuracy; see Table \ref{tab:hera_error}.  

The comparison of the running times of our approaches with that of Hera (for a range of different approximation factors) can be found in Figure \ref{fig:runtime}, which summarizes the runtime for each method on a {\bf log-scale} using both randomly generated diagrams as well as the reddit-binary dataset (real persistence diagrams from graphs). We compare the speed of our modified flowtree and $L_1$ embedding against Hera where the parameter $\eps$ for Hera is set to be 300000. However, note that as one relaxes the approximation parameter $\eps$, the speed of Hera in fact does not improve much (as shown in Figure \ref{fig:runtime}). Thus our speed gains remain no matter which choice of $\eps$ we use for Hera. Additionally, the true approximation error for Hera also does not decrease much; see Table \ref{tab:hera_error}. To get the approximation error for Hera, we find the true Wasserstein distance by using the $\texttt{wasserstein\_distance}$ function from the GUDHI library which uses the Python Optimal Transport library \cite{flamary2017pot} and is based on ideas from \cite{lacombe2018large}.
The results in Figure \ref{fig:runtime} indicate that the modified flowtree approach is between $50$ and $1000$ times faster than Hera and the difference increases as the size of the diagrams increases. Similarly, the $L_1$ embedding approach is between $150$ and $4900$ times faster than Hera. Both are order of magnitudes faster than Hera; but the price to pay is that the approximation factor is worse for our approach as shown in Appendix \ref{appendix:results} Figure \ref{fig:vp_error}. 
%In practice, we believe that Hera's output is very close to ground truth, and adjusting the parameter does not change much practical performance neither in terms of speed or quality. 

\subparagraph*{Approximation accuracy comparison.} 
To measure the accuracy of the approximation of both the modified flowtree and $L_1$ embedding approach, we first measure the average relative error and standard deviation of both methods on all datasets using $L_1$, $L_2$, and $L_\infty$ as the ground metrics. In particular, given a ground truth distance $d$ and an approximate distance $\tilde{d}$, the relative error is $\rho(d) = \frac{|d - \tilde{d}|}{d}$. 
As mentioned earlier, we use the output of Hera for $\eps=0.01$ as ground truth, and compare our approximated distance with that. The results are summarized in Figure \ref{fig:vp_error} and a detailed table of the average relative error and standard deviation is in Table \ref{tab:average error}. 
Overall, while our modified flowtree is slower than the $L_1$-embedding approach, it achieves much better approximation error. 
For our experiments, we generate a quadtree only once for all persistence diagrams in a given dataset and calculate error for the approximated distances for the single quadtree. However, note that by constructing several quadtrees and averaging the distance estimates or taking the smallest estimated, we could potentially reduce the approximation error. 
% Note that for our code, we generate the quadtree once using points for all relevant persistence diagrams and use the algorithms above on the tree. By constructing several quadtrees - we could potentially cut down on error by averaging or taking the smallest estimate. However, for the error here is taken by constructing a single quadtree. 

In addition to the average error of our approximate distances, we can also consider the efficacy of both methods in terms of nearest neighbor search and ranking accuracy. To evaluate nearest neighbor search, we first split the set of persistence diagrams into query diagrams and candidate diagrams. Then, we measure recall@$m$ accuracy where recall@$m$ is defined as the fraction of queries that have the true nearest neighbor within the top $m$-ranked candidates returned by the evaluated method. The results are reported in Figure \ref{fig:recall@m}. For ranking accuracy, the detailed results are in Appendix \ref{appendix:results}. To summarize, the modified flowtree approach is more accurate than the $L_1$ embedding approach both in terms of nearest neighbor search and closeness to the true ranking of candidate diagrams for a fixed query diagram. Both approaches appear to have a lower degree of accuracy %(in relation to nearest neighbor search and ranking accuracy) 
on diagrams where there a higher proportion of points near the diagonal. This may be due to the increased possibility of erroneously matching points to the diagonal.

In summary, we note that both our new approximation algorithms significantly improve the speed previously best-known Hera algorithm by orders of magnitudes, but with worse approximation factors. Empirically, the approximation factors remain constant despite our theoretical results suggestion an $O(\log \Delta)$-approximation. In particular, the relative approximation error of flowtree is often smaller than 0.50 for $L_2$ ground metric (see Appendix \ref{appendix:results} Table \ref{tab:average error}). 

% \subsection{Speed comparison}
% Note that Hera was reported to be 50 times faster than Dionysus, an implementation of the Hungarian algorithm, for small instances to around 400 times faster than Dionysus for large instances \cite{Hera}. Additionally, average relative error for Hera eventually plateaus at 0.0076 despite increases in the maximum allowed relative error. Refer to Table \ref{tab:hera_error} for more details. Since we cannot relax the average relative error for Hera to be comparable with the modified flowtree, we compare both the speed of the modified flowtree and $L_1$ embedding against Hera where the average relative error is 0.0076 using both randomly generated diagrams as well as the reddit-binary dataset (real persistence diagrams from graphs). Figure \ref{fig:runtime} summarizes the runtime for each method on a log-scale. The results in Figure \ref{fig:runtime} indicate that the modified flowtree approach is between 50 and 1000 times faster than Hera and the difference increases as the size of the diagrams increases. Similarly, the $L_1$ embedding approach is between 150 and 4900 times faster than Hera. 
\begin{table}[ht]
    \centering
    \begin{tabular}{|c|c|}
    \hline
        Maximum allowed relative error & Average relative error\\
        \hline
        0.1 & 0.00043768\\
        \hline
        1.0 & 0.003331\\
        \hline
        100000 & 0.0076055\\
        \hline
        200000 & 0.0076055\\
        \hline
    \end{tabular}
    \caption{Comparison of the maximum allowed relative error with the average experimental relative error for Hera on the reddit-binary dataset. The relative error was calculated using the GUDHI library's \texttt{wasserstein\_distance} function which uses the Python Optimal Transport library to compute exact 1-Wasserstein distance. }
    \label{tab:hera_error}
\end{table}
\begin{table}[h]
    \centering
    \begin{tabular}{|c|c|c | c | c |c|c|c|}
    \hline
        & & \multicolumn{2}{c|}{$L_1$} & \multicolumn{2}{c|}{$L_2$} & \multicolumn{2}{c|}{$L_\infty$}\\
        \cline{3-8}
        &  & embd & ft & embd & ft & embd & ft \\
         \cline{1-8}
       \multirow{2}{*}{synthetic-uniform} &  Avg. Error & 2.058 & 0.2846 & 3.161 & 0.2664 & 4.536 & 0.2595\\
         \cline{2-8}
        &  Std. Dev. & 1.034 & 0.3891 & 1.189 & 0.3488 & 1.164 & 0.3176\\
         \hline
         \multirow{2}{*}{synthetic-Gaussian}& Avg. Error  & 1.341 & 0.3358 & 2.136 & 0.2860 & 3.035 & 0.2251\\
         \cline{2-8}
         & Std. Dev. & 0.3647 & 0.1809 & 0.4139 & 0.1519 & 0.3669 & 0.1178\\
         \hline
         \multirow{2}{*}{reddit-binary} &  Avg. Error & 2.112 & 0.2899 & 3.089 & 0.3080 & 3.921 & 0.2854\\
         \cline{2-8}
         & Std. Dev. & 1.275 & 0.3859 & 2.100 & 0.5126 & 2.427 & 0.4801\\
         \hline
         \multirow{2}{*}{ModelNet10} & Avg. Error &  2.189 & 0.7331 & 2.438 & 0.4929 & 3.061 & 0.9399\\
         \cline{2-8}
         &Std. Dev. &  0.9543 & 0.4132 & 0.8136 & 0.3171 & 1.051 & 0.4448\\
         \hline
    \end{tabular}

    \caption{Average error and standard deviation for all datasets. We abbreviate the $L_1$ embedding approach to embd and the flowtree approach to ft.}
    \label{tab:average error}
\end{table}

\begin{figure}
    \centering
    \begin{subfigure}{0.45\textwidth}
      \centering
      \includegraphics[width=1.1\textwidth]{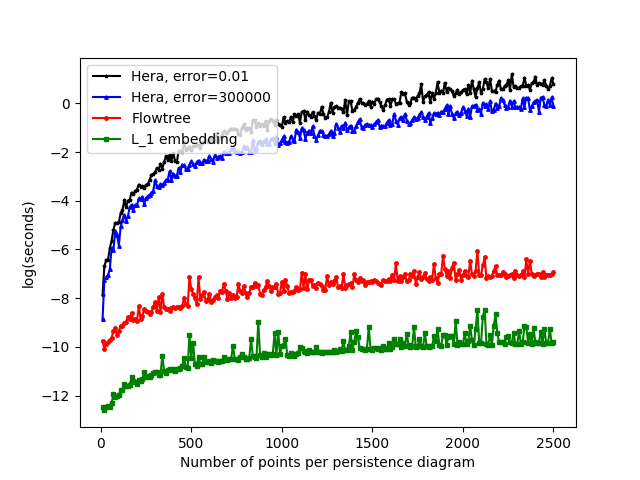}
      \subcaption{synthetic-uniform.}
    \end{subfigure}%
    \hfill
    \begin{subfigure}{0.45\textwidth}
      \centering
       \includegraphics[width=1.1\textwidth]{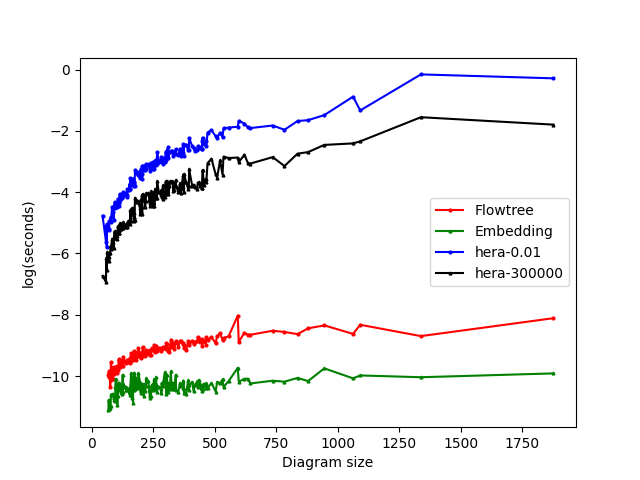}
       \subcaption{reddit-binary.}
    \end{subfigure}
    \caption{Comparison of the runtimes of HERA, flowtree, and $L_1$ embedding using both generated and real data with $L_2$ as the ground metric.}
    \label{fig:runtime}
\end{figure}

\begin{figure}
    \centering
    \begin{subfigure}{0.45\textwidth}
      \centering
      \includegraphics[width=1.1\textwidth]{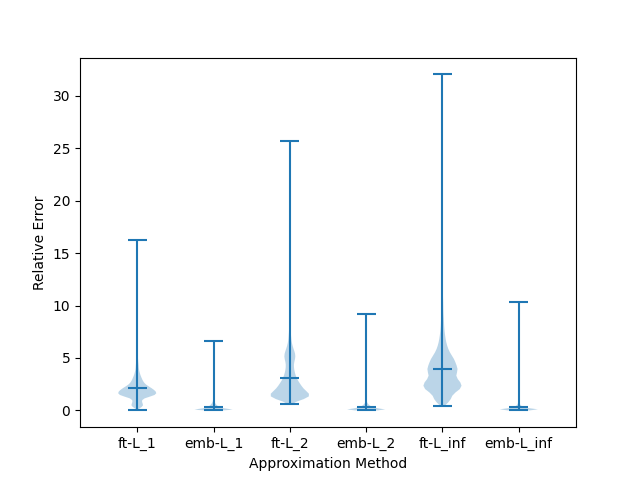}
      \subcaption{reddit-binary.}
    \end{subfigure}%
    \hfill
    \begin{subfigure}{0.45\textwidth}
      \centering
       \includegraphics[width=1.1\textwidth]{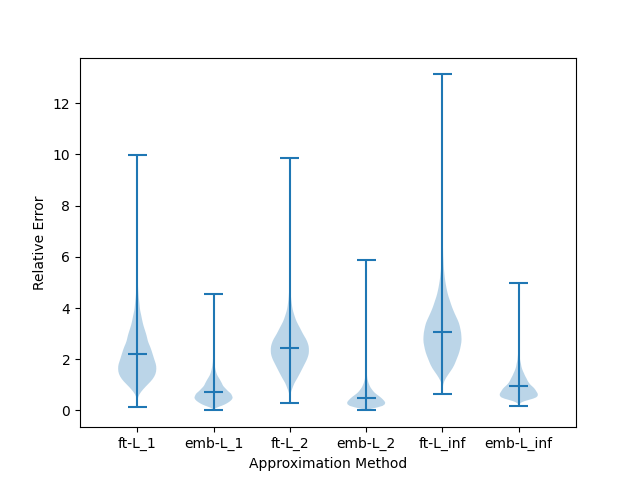}
       \subcaption{ModelNet10.}
    \end{subfigure}
        \vskip\baselineskip
    \begin{subfigure}{0.45\textwidth}
      \centering
      \includegraphics[width=1.1\textwidth]{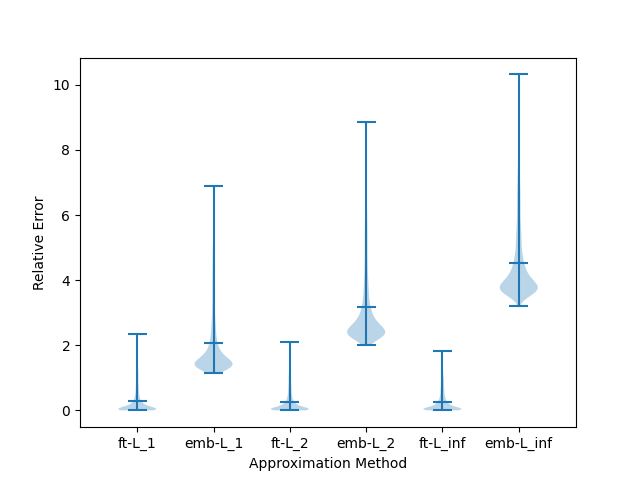}
      \subcaption{synthetic-uniform.}
    \end{subfigure}%
    \hfill
   \begin{subfigure}{0.45\textwidth}
      \centering
      \includegraphics[width=1.1\textwidth]{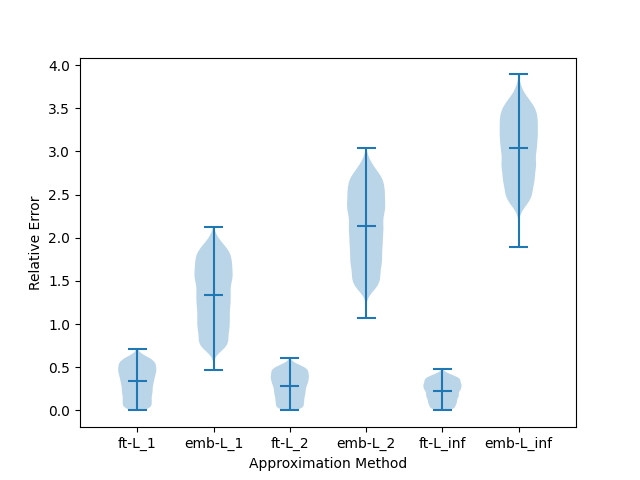}
      \subcaption{synthetic-Gaussian.}
    \end{subfigure}%
    \caption{Relative error for flowtree and quadtree approximations over $L_1$, $L_2$, and $L_\infty$ ground metrics for all datasets.}
     \label{fig:vp_error}
\end{figure}
\begin{figure}[htpb]
\centering
\begin{subfigure}{.45\textwidth}
  \centering
  \includegraphics[width=1.1\linewidth]{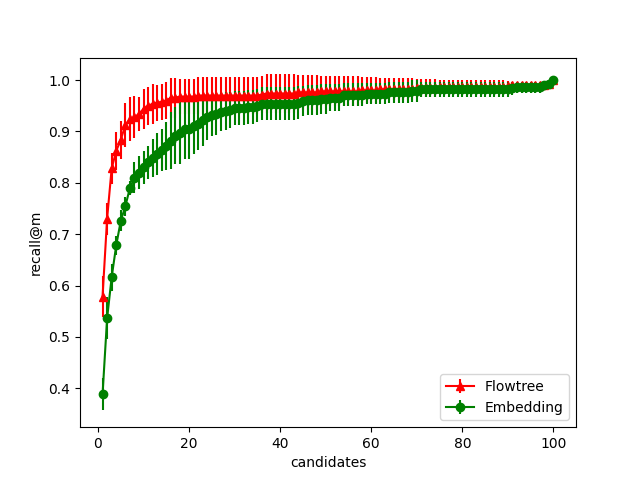}
  \caption{Recall@$m$ accuracy on reddit-binary dataset with $L_2$ ground metric.}
\end{subfigure}%
\hfill
\begin{subfigure}{.45\textwidth}
  \centering
  \includegraphics[width=1.1\linewidth]{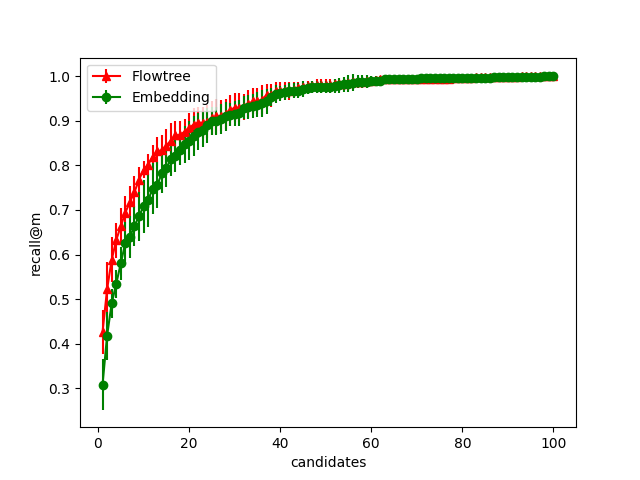}
  \caption{Recall@$m$ accuracy ModelNet10 dataset with $L_2$ ground metric.}
\end{subfigure}
\vskip\baselineskip
\begin{subfigure}{.45\textwidth}
  \centering
  \includegraphics[width=1.1\linewidth]{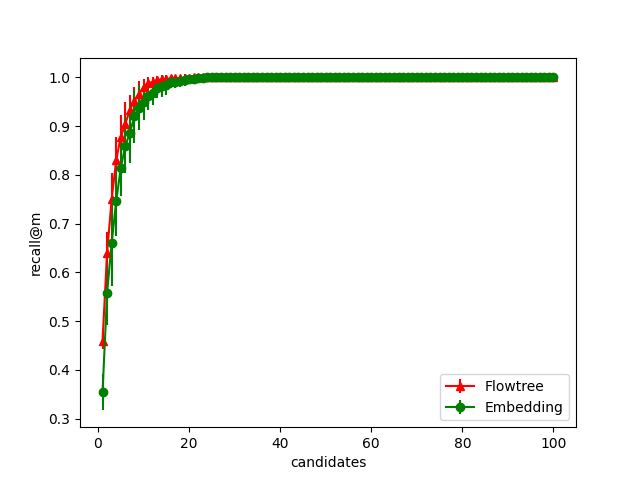}
  \caption{Recall@$m$ accuracy on synthetic-uniform dataset with $L_2$ ground metric}
\end{subfigure}%
\hfill
\begin{subfigure}{.45\textwidth}
  \centering
  \includegraphics[width=1.1\linewidth]{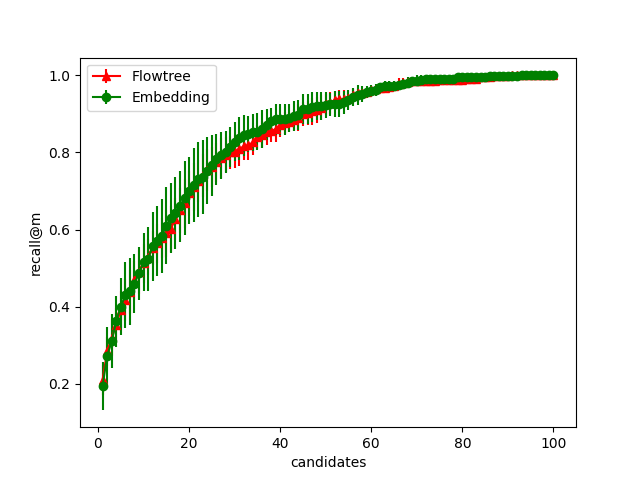}
  \caption{Recall@$m$ accuracy on synthetic-Gaussian dataset with $L_2$ ground metric}
\end{subfigure}%
\caption{Recall@$m$ accuracy on reddit-binary and ModelNet10 datasets with $L_2$ ground metric.}
\label{fig:recall@m}
\end{figure}

\newpage
\section{Concluding remarks}
In this paper, we presented two algorithms for fast approximation of the 1-Wasserstein distance between persistence diagrams based on $L_1$ embedding. While the relative error incurred by both algorithms is higher than that of Hera, the runtime is significantly faster. We also observe that approximation methods introduced are more accurate on persistence diagrams with a lower proportion of points near the diagonal. 

In the future, we are interested in using the $L_1$ embedding described with locality sensitive hashing for sub-linear nearest neighbor search. Additionally, it maybe be possible to use the ideas to compare persistence diagrams under some transformations: e.g, parallel shifting along the diagonal directions (which corresponding to that the input functions generating the persistence diagram is added by a constant term). It will also be interesting to expand this work to perform statistics on the space of persistence diagrams (e.g, computing 1-mean of persistence diagrams under our approximation distances).

%%
%% Bibliography
%%

%% Please use bibtex, 

\newpage
\bibliography{refs}

\begin{thebibliography}{10}

\bibitem{Flowtree}
Arturs Backurs, Yihe Dong, Piotr Indyk, Ilya Razenshteyn, and Tal Wagner.
\newblock Scalable nearest neighbor search for optimal transport.
\newblock {\em arXiv preprint arXiv:1910.04126}, 2019.

\bibitem{Bartal96}
Yair Bartal.
\newblock Probabilistic approximation of metric spaces and its algorithmic
  applications.
\newblock In {\em Proceedings of 37th Conference on Foundations of Computer
  Science}, pages 184--193. IEEE, 1996.

\bibitem{Bartal98}
Yair Bartal.
\newblock On approximating arbitrary metrices by tree metrics.
\newblock {\em STOC}, 98:161--168, 1998.

\bibitem{BertsekasAuction}
Dimitri Bertsekas.
\newblock The auction algorithm: A distributed relaxation method for the
  assignment problem.
\newblock {\em Annals of Operations Research}, 14(1):105--123, 1988.

\bibitem{Buchet2018}
Micka{\"e}l Buchet, Yasuaki Hiraoka, and Ippei Obayashi.
\newblock Persistence homology and material and informatics.
\newblock In Isao Tanaka, editor, {\em Nanoinformatics}, pages 75--95. Spring
  Singapore, Singapore, 2018.
\newblock \href {https://doi.org/10.1007/978-981-10-7617-6_5}
  {\path{doi:10.1007/978-981-10-7617-6_5}}.

\bibitem{real_datasets}
Chen Cai and Yusu Wang.
\newblock Understanding the power of persistence pairing via permutation test.
\newblock {\em arXiv preprint arXiv:2001.06058}, 2020.

\bibitem{Charikar2002}
Moses Charikar.
\newblock Similarity estimation techniques from rounding algorithms.
\newblock In {\em Proceedings of the thirty-fourth annual ACM symposium on
  theory of computing}, pages 2292--2300. ACM, 2002.

\bibitem{CCGMO09}
Fr{\'e}d{\'e}ric Chazal, David Cohen-Steiner, Leonidas~J. Guibas, Facundo
  M\'{e}moli, and Steve~Y. Oudot.
\newblock Gromov-hausdorff stable signatures for shapes using persistence.
\newblock {\em Computer Graphics Forum}, 28(5):1393--1403, 2009.

\bibitem{PersistenceStability2007}
David Cohen-Steiner, Herbert Edelsbrunner, and John Harer.
\newblock Stability of persistence diagrams.
\newblock {\em Discrete Comput. Geom.}, 37(1):103--120, 2007.

\bibitem{EdelsbrunnerHarer2010}
Herbert Edelsbrunner and John Harer.
\newblock {\em Computational topology: an introduction}.
\newblock American Mathematical Society, 2010.

\bibitem{EdelsLetsZomo2002}
Herbert Edelsbrunner, David Letscher, and Afra Zomorodian.
\newblock Topological persistence and simplification.
\newblock {\em Discrete Comput. Gemo.}, 28:511--533, 2002.

\bibitem{flamary2017pot}
R\'emi Flamary and Nicolas Courty.
\newblock Pot: Python optimal transport library, 2017.
\newblock URL: \url{https://pythonot.github.io/}.

\bibitem{GambleHo2010}
Jennifer Gamble and Giseon Ho.
\newblock Exploring uses of persistence homology for statistical analysis of
  landmark-based shape data.
\newblock {\em Journal of Multivariate Analysis}, 101(9):2184--2199, 2010.

\bibitem{Hiraoka7035}
Yasuaki Hiraoka, Takenobu Nakamura, Akihiko Hirata, Emerson~G. Escolar, Kaname
  Matsue, and Yasumasa Nishiura.
\newblock Hierarchical structures of amorphous solids characterized by
  persistent homology.
\newblock {\em Proceedings of the National Academy of Sciences},
  113(26):7035--7040, 2016.
\newblock URL: \url{https://www.pnas.org/content/113/26/7035}, \href
  {http://arxiv.org/abs/https://www.pnas.org/content/113/26/7035.full.pdf}
  {\path{arXiv:https://www.pnas.org/content/113/26/7035.full.pdf}}, \href
  {https://doi.org/10.1073/pnas.1520877113}
  {\path{doi:10.1073/pnas.1520877113}}.

\bibitem{IndykThaper2003}
Piotr Indyk and Nitin Thaper.
\newblock Fast image retrieval via embeddings.
\newblock In {\em 3rd international workshop on statistical and computational
  theories of vision}, volume~2, page~5, 2003.

\bibitem{Kalantari95}
Bahman Kalantari and Iraj Kalantari.
\newblock A linear-time algorithm for minimum cost flow on undirected
  one-trees.
\newblock {\em Combinatorics Advances}, pages 217--223, 1995.

\bibitem{HessNeuron}
Lida Kanari, Pawe\l{} D{\l}otko, Martina Scolamiero, Ran Levi, Julian
  Shillcock, Kathryn Hess, and Henry Markram.
\newblock A topological representation of branching neuronal morphologies.
\newblock {\em Neuroinformatics}, 16(1):3--13, 2018.
\newblock \href {https://doi.org/10.1007/s12021-017-9341-1}
  {\path{doi:10.1007/s12021-017-9341-1}}.

\bibitem{Hera}
Michael Kerber, Dimitriy Morozov, and Arnur Nigmetov.
\newblock Geometry helps to compare persistence diagrams.
\newblock {\em Journal of Experimental Algorithmics(JEA)}, 22:1--20, 2017.

\bibitem{KleinbergTardos2002}
Jon Kleinberg and Eva Tardos.
\newblock Approximation algorithms for classification problems with pairwise
  relationships: Metric labeling and markov random fields.
\newblock {\em Journal of the ACM (JACM)}, 49(5):616--639, 2002.

\bibitem{Kuhn}
Harold Kuhn.
\newblock The hungarian method for the assignment problem.
\newblock {\em Naval research logistics quarterly}, 2(1-2):83--97, 1955.

\bibitem{lacombe2018large}
Théo Lacombe, Marco Cuturi, and Steve Oudot.
\newblock Large scale computation of means and clusters for persistence
  diagrams using optimal transport, 2018.
\newblock \href {http://arxiv.org/abs/1805.08331} {\path{arXiv:1805.08331}}.

\bibitem{LYFC19}
Tam Le, Makoto Yamada, Kenji Fukumizu, and Marco Cuturi.
\newblock Tree-sliced approximation of wasserstein distances.
\newblock {\em arXiv preprint arXiv:1902.00342}, 2019.

\bibitem{LB17}
Yongjin Lee, Senja~D. Barthel, Pawe\l{} D\l{}otko, S.~Mohamad Moosavi, Kathryn
  Hess, and Berend Smit.
\newblock Quantifying similarity of pore-geometry in nanoporous materials.
\newblock {\em Nat Commun.}, 8:15396, 2017.
\newblock \href {https://doi.org/https://doi.org/10.1038/ncomms15396}
  {\path{doi:https://doi.org/10.1038/ncomms15396}}.

\bibitem{LWAPW17}
Yanjie Li, Dingkang Wang, Giorgio Ascoli, Partha Mitra, and Yusu Wang.
\newblock Metrics for comparing neuronal tree shapes based on persistent
  homology.
\newblock {\em PLOS One}, 12(8):1--24, 2017.
\newblock \href {https://doi.org/10.1371/journal.pone.0182184}
  {\path{doi:10.1371/journal.pone.0182184}}.

\bibitem{gudhi}
Cl\'ement Maria, Jean-Daniel Boissonat, Marc Glisse, and Mariette Yvinec.
\newblock The gudhi library: simplicial complexes and persistent homology,
  2014.
\newblock URL: \url{http://gudhi.gforge.inria.fr/python/latest/index.html}.

\bibitem{Dionysus}
Dimitriy Morozov.
\newblock Dionysus, 2010.
\newblock URL: \url{mrzv.org/software/dionysus}.

\bibitem{lfm-pro2007}
Ahmet Sacan, Ozgur Ozturk, Hakan Ferhatosmanoglu, and Yusu Wang.
\newblock Lfm-pro: a tool for detecting significant local structural sites in
  proteins.
\newblock {\em Bioinformatics}, 23(6):709--716, 2007.

\bibitem{SatoYamada2020}
Ryoma Sato, Makoto Yamada, and Hisashi Kashima.
\newblock Fast unbalanced optimal transport on tree.
\newblock {\em arXiv preprint arXiv:2006.02703}, 2020.

\bibitem{SOC10}
Primoz {Skraba}, Maks {Ovsjanikov}, Fr\'ed\'eric {Chazal}, and Leonidas
  {Guibas}.
\newblock Persistence-based segmentation of deformable shapes.
\newblock In {\em 2010 IEEE Computer Society Conference on Computer Vision and
  Pattern Recognition - Workshops}, pages 45--52, 2010.
\newblock \href {https://doi.org/10.1109/CVPRW.2010.5543285}
  {\path{doi:10.1109/CVPRW.2010.5543285}}.

\bibitem{ModelNet}
Zhirong Wu, Shuran Song, Aditya Khosla, Fisher Yu, Linguang Zhang, Xiaoou Tang,
  and Jianxiong Xiao.
\newblock 3d shapenets: A deep representation for volumentric shapes.
\newblock In {\em Proceedings of the IEEE conference on computer vision and
  pattern recognition}, pages 1912--1920, 2015.

\end{thebibliography}

\newpage
\appendix
\section{Additional Proofs}
\label{appendix:proofs}
\begin{proof}[Proof of observation \ref{ot_w1}]
By the optimality of $\dOT(\mu_{\widehat{\perD}}, \nu_{\widehat{\anotherperD}})$, we know that
\begin{equation*}
    \dOT(\mu_{\widehat{\perD}}, \nu_{\widehat{\anotherperD}}) \leq \dW(P, Q) + \sum_{(a, b) \in \Gamma_1} ||\pi(a) - \pi(b)||_q
\end{equation*}
where $\Gamma_1$ is the set of $(a, b) \in \Gamma$ that have first form given in definition \ref{def:dWperDV2}. We know that $||\pi(a) - \pi(b)||_q \leq ||a - b||_q$ so $\dOT(\mu_{\widehat{\perD}}, \nu_{\widehat{\anotherperD}})\leq \dW(P, Q) + \sum_{(a, b) \in \Gamma_1} ||\pi(a) - \pi(b)||_q \leq 2 \cdot \dW(P, Q)$. 
\end{proof}
\section{Datasets} 
\label{appendix:datasets}
We use datasets of synthetic persistence diagrams as well as persistence diagrams generated from real data. For persistence diagrams from real data, we use both graph and shape datasets.
\begin{itemize}
    \item \textit{Synthetic data}: We generated two sets of persistence diagrams. For the first set of synthetic persistence diagrams, we find a random persistence of size at most $s$ where $s = 10x$ for $x \in \{1, \dots, 100\}$ by generating points $p_1, \dots, p_s$. To find each points $p_i$, we sample $p_i.x$ from a uniform distribution from $0$ to $200$. We then sample $p_i.y$ from a uniform distribution between $x$ and $300$. We will refer to this set of synthetic persistence diagrams as synthetic-uniform. For the second set of synthetic persistence diagrams, we again generate a point $p_i$ by sampling $p_i.x$ from a uniform distribution from $0$ to $200$. We then sample $p_i.y$ from a Gaussian distribution centered about $x$ with standard deviation 1.0. We will refer to the second set of synthetic diagrams as synthetic-Gaussian. 
    
    \item \textit{Graphs}: For persistence diagrams generated from graphs, we use the reddit-binary graph dataset, which consists of graphs corresponding to online discussion on Reddit. In each graph, nodes correspond to users and there is an edge between two nodes if at least one of the corresponding users has responded to the other's comments. The data is taken from four popular subreddits: IAmA, AskReddit, TrollXChromosomes, and atheism. Additionally, the persistence diagrams are generated using node degree as the filtration function \cite{real_datasets}.
    
    \item \textit{Shapes}: We use the ModelNet10 \cite{ModelNet} dataset to generate persistence diagrams from shapes. ModelNet10 is comprised of 4899 CAD models from 10 object categories. The persistence diagrams are generated using closeness centrality as the filtration function \cite{real_datasets}. 
    
\end{itemize}
The statistics of the ModelNet10 and reddit-binary datasets are summarized in Table \ref{tab:diagram_statistics}. Note that a persistence point is considered close to the diagonal if its lifetime is less than one-tenth of the lifetime of the point with the largest lifetime.  

\begin{table}[ht]
    \centering
    \begin{tabular}{|c|c|c|}
        \hline
         &  Average \# of PD points &  Average \# of points near diagonal\\
        \hline
        reddit-binary & 278.155 & 20.245 \\
        \hline
        ModelNet10 & 40.52 & 34.345\\
        \hline
    \end{tabular}
    \caption{Diagram statistics for reddit-binary and ModelNet10 datasets.}
    \label{tab:diagram_statistics}
\end{table}

\section{Results}
\label{appendix:results}
To measure the accuracy of the rankings produced by the modified $L_1$ embedding and flowtree methods, we plot the true ranking of each candidate against the rank of the candidate in the rankings produced by the evaluated method. Note that we will do this using the $L_2$ norm as the ground metric. The ranking accuracies for the evaluated methods for the reddit-binary dataset and ModelNet10 are summarized in Figures \ref{fig:rankings_rbin}, \ref{fig:rankings_mn10}, \ref{fig:rankings_gaussian}, \ref{fig:rankings_unif}. The average number of ranks away from the true rank is less than 10 for both reddit-binary and synthetic-uniform whereas the same metric for synthetic-uniform and ModelNet10 is above ten for both datasets. 

Both the modified flowtree and the $L_1$ embedding approximations seem to be less effective for estimating nearest neighbor for persistence diagrams where there is a high proportion of points near the diagonal. This may be because a higher proportion of points near the diagonal increases the possibility of erroneously matching points to the diagonal. 

\begin{figure}[h]
    \centering
    \begin{subfigure}{.45\textwidth}
      \centering
      \includegraphics[width=1.1\linewidth]{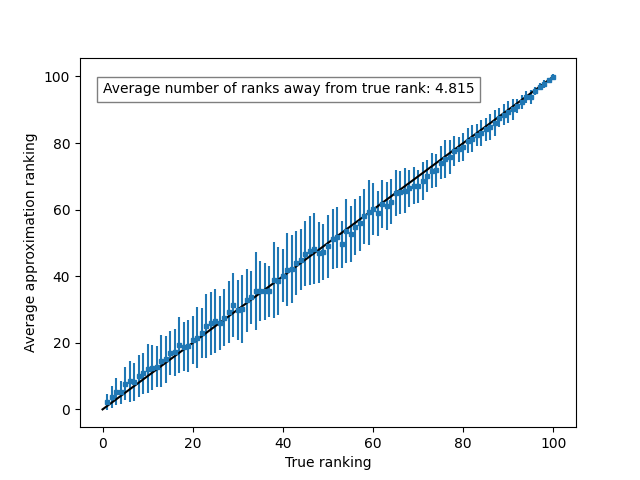}
      \caption{Flowtree rankings.}
    \end{subfigure}%
    \hfill
    \begin{subfigure}{.45\textwidth}
      \centering
      \includegraphics[width=1.1\linewidth]{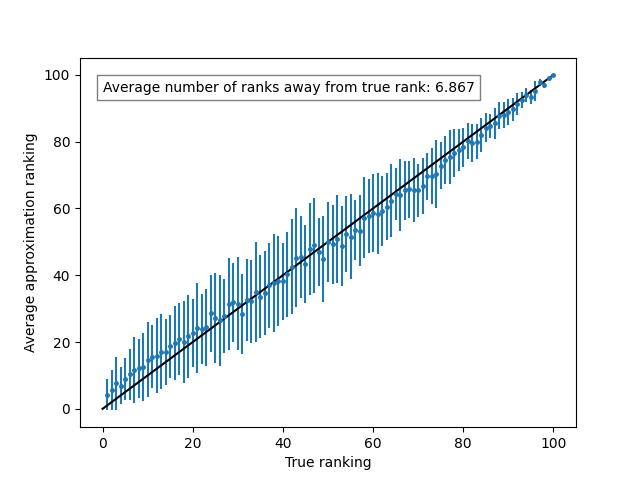}
      \caption{$L_1$ embedding rankings.}
    \end{subfigure}
    \caption{Comparison of rankings generated by the flowtree and $L_1$ embedding approximations with the true rankings of the candidate diagrams using the reddit-binary datset.}
    \label{fig:rankings_rbin}
\end{figure}

\begin{figure}
    \centering
    \begin{subfigure}{.45\textwidth}
      \centering
      \includegraphics[width=1.1\linewidth]{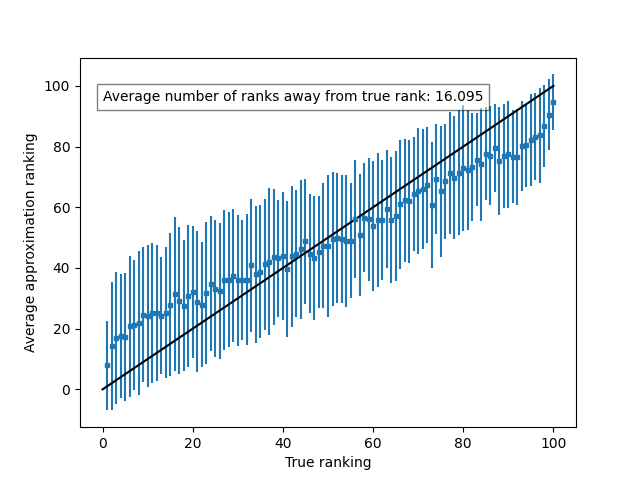}
      \caption{Flowtree rankings.}
    \end{subfigure}%
    \hfill
    \begin{subfigure}{.45\textwidth}
      \centering
      \includegraphics[width=1.1\linewidth]{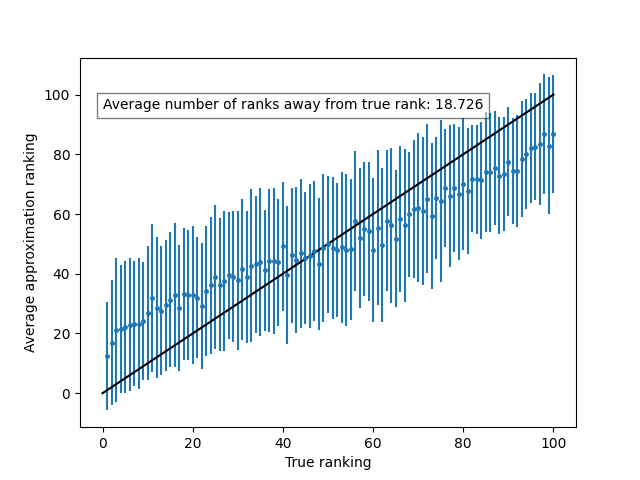}
      \caption{$L_1$ embedding rankings.}
    \end{subfigure}
    \caption{Comparison of rankings generated by the flowtree and $L_1$ embedding approximations with the true rankings of the candidate diagrams using the ModelNet10.}
    \label{fig:rankings_mn10}
\end{figure}

\begin{figure}
    \centering
    \begin{subfigure}{.45\textwidth}
      \centering
      \includegraphics[width=1.1\linewidth]{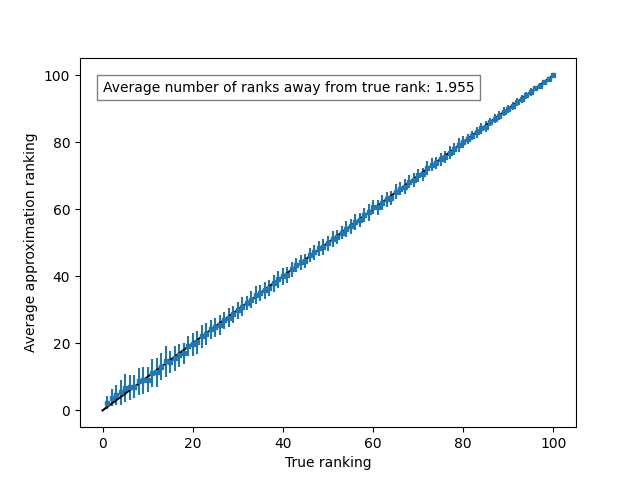}
      \caption{Modified flowtree rankings.}
    \end{subfigure}%
    \hfill
    \begin{subfigure}{.45\textwidth}
      \centering
      \includegraphics[width=1.1\linewidth]{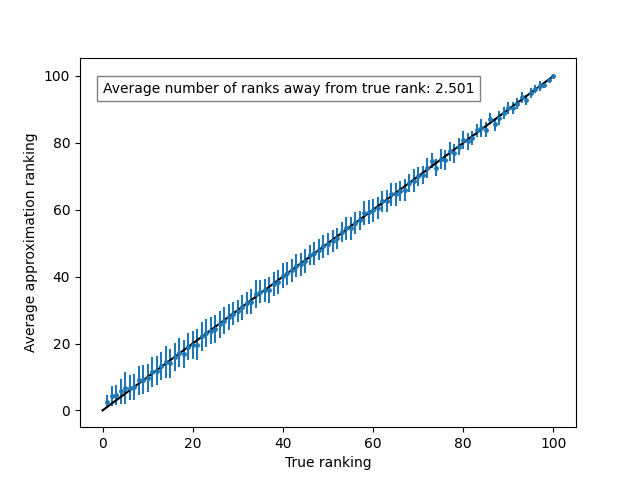}
      \caption{$L_1$ embedding rankings.}
    \end{subfigure}
    \caption{Comparison of rankings generated by the modified flowtree and $L_1$ embedding approximations with the true rankings of the candidate diagrams using synthetic-uniform dataset.}
    \label{fig:rankings_unif}
\end{figure}

\begin{figure}[h]
    \centering
    \begin{subfigure}{.45\textwidth}
      \centering
      \includegraphics[width=1.1\linewidth]{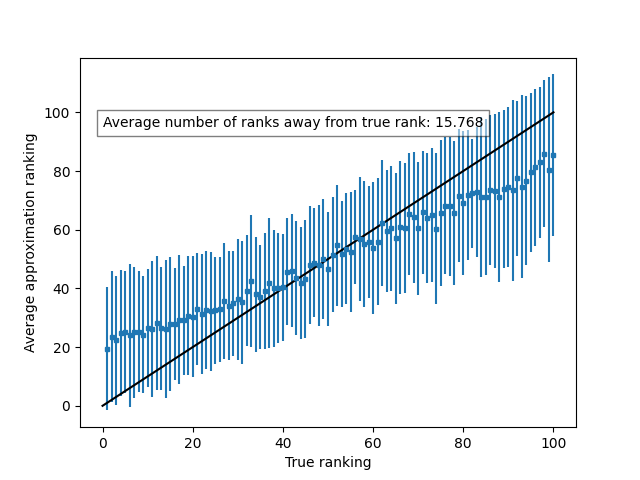}
      \caption{Modified flowtree rankings.}
    \end{subfigure}%
    \hfill
    \begin{subfigure}{.45\textwidth}
      \centering
      \includegraphics[width=1.1\linewidth]{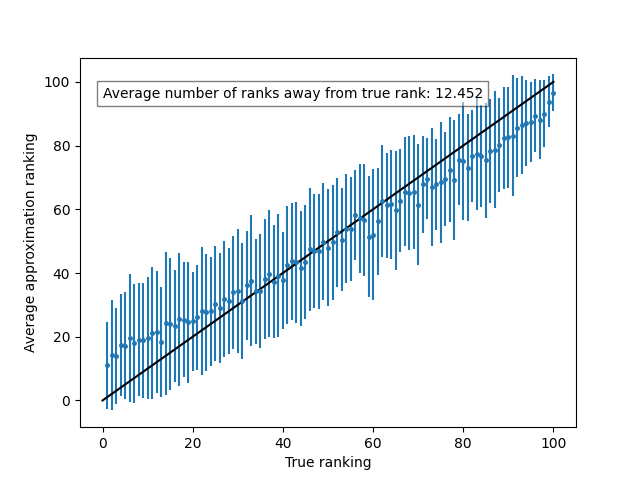}
      \caption{$L_1$ embedding rankings.}
    \end{subfigure}
    \caption{Comparison of rankings generated by the modified flowtree and $L_1$ embedding approximations with the true rankings of the candidate diagrams using synthetic-Gaussian.}
    \label{fig:rankings_gaussian}
\end{figure}

\end{document}